\documentclass{article}

\usepackage{subfigure}
\usepackage{float}
\usepackage{listings}
\usepackage{graphicx}
\usepackage{xurl}
\PassOptionsToPackage{hyphens}{url}\usepackage{hyperref}
\usepackage{orcidlink}
\usepackage{doi}
\usepackage{tabularx}
\usepackage{tabulary}
\usepackage{booktabs}
\usepackage{multirow}
\usepackage{hyphenat}
\usepackage{titlesec}
\usepackage{enumitem}
\usepackage[ruled,linesnumbered]{algorithm2e}
\usepackage{csquotes}
\usepackage{tikz}
\newcommand*\circled[1]{\tikz[baseline=(char.base)]{
            \node[shape=circle,draw,inner sep=0pt] (char) {#1};}}
\newcommand*\circledzero[1]{\tikz[baseline=(char.base)]{
            \node[shape=circle,draw,inner sep=0.6pt] (char) {#1};}}
\usepackage{mathtools}

\usepackage{url}
\usepackage{tablefootnote}
\newcommand{\name}[1]{DEVA}
\usepackage{amsthm}
\newtheorem{theorem}{Theorem}
\usepackage{authblk}

\hypersetup{
    colorlinks,
    linkcolor={red!60!black},
    citecolor={blue!80!black},
    urlcolor={green!50!black}
}
\usepackage{textcomp}
\newcommand{\myapprox}{\raisebox{0.5ex}{\texttildelow}}

\newcommand\blfootnote[1]{%
  \begingroup
  \renewcommand\thefootnote{}\footnote{#1}%
  \addtocounter{footnote}{-1}%
  \endgroup
}

\begin{document}

\title{In-Vehicle Edge System for Real-Time Dashcam Video Analysis\footnote{This work was supported by the National Research Foundation of Korea (NRF) grants funded by the Korea government (MSIT) (No. 2023R1A2C1003984)}}

\author[a]{Seyul Lee}
\author[b]{Jayden King}
\author[b]{Young Choon Lee}
\author[c]{Hyuck Han}
\author[a]{Sooyong Kang}
\affil[a]{\small{Dept. of Computer Science, Hayang University, Seoul, South Korea}}
\affil[b]{\small{School of Computing, Macquarie University, Sydney, Australia}}
\affil[c]{\small{Dept. of Computer Science, Dongduk Women's University, Seoul, South Korea}}

\date{}
\maketitle

\blfootnote{
    \hypersetup{hidelinks}\emph{Email addresses}:
    \href{mailto:djm03178@hanyang.ac.kr}{\texttt{djm03178@hanyang.ac.kr}} (Seyul~Lee),
    \href{mailto:jayden.king@mq.edu.au}{\texttt{jayden.king@mq.edu.au}} (Jayden~King),
    \href{mailto:young.lee@mq.edu.au}{\texttt{young.lee@mq.edu.au}} (Young~Choon~Lee),
    \href{mailto:hhyuck96@dongduk.ac.kr}{\texttt{hhyuck96@dongduk.ac.kr}} (Hyuck~Han),
    \href{mailto:sykang@hanyang.ac.kr}{\texttt{sykang@hanyang.ac.kr}} (Sooyong~Kang)
}

\begin{abstract}
Modern vehicles equip dashcams that primarily collect visual evidence for traffic accidents. However, most of the video data collected by dashcams that is not related to traffic accidents is discarded without any use. In this paper, we present a use case for dashcam videos that aims to improve driving safety. By analyzing the real-time videos captured by dashcams, we can detect driving hazards and driver distractedness to alert the driver immediately. To that end, we design and implement a Distributed Edge-based dashcam Video Analytics system (\name{}), that analyzes dashcam videos using personal edge (mobile) devices in a vehicle. \name{} consolidates available in-vehicle edge devices to maintain the resource pool, distributes video frames for analysis to devices considering resource availability in each device, and dynamically adjusts frame rates of dashcams to control the overall workloads. The entire video analytics task is divided into multiple independent phases and executed in a pipelined manner to improve the overall frame processing throughput. We implement \name{} in an Android app and also develop a dashcam emulation app to be used in vehicles that are not equipped with dashcams. Experimental results using the apps and commercial smartphones show that \name{} can process real-time videos from two dashcams with frame rates of around 22\myapprox30 FPS per camera within 200 ms of latency, using three high-end devices.
\end{abstract}

\section{Introduction}\label{sec:introduction}

Dashcams are being widely used in vehicles nowadays to collect evidence for traffic accidents. Yet the vast majority of dashcam-produced video data does not capture traffic accidents and is discarded. Instead of wasting this video data, further value may be gained from them through video analytics such as real-time object detection and tracking for safe driving. However, real-time video analytics is a compute-intensive task that requires a large amount of computational resources which are not available in current commercial dashcams.

Offloading video analytics tasks to the cloud can be an alternative to in-situ processing in terms of resource requirements. However, it is not suitable when the real-time constraint is enforced on the video analytics tasks, as in object detection and tracking applications for safe driving, due to two reasons. First, the large and unstable communication latency between vehicles and cloud data centers makes it difficult to provide timely analysis of continuously generating video frames. Second, the huge amount of video data simultaneously generated from a large number of driving vehicles may flood the internet, affecting the quality of not only the video analytics services but also other internet services.

These issues can be resolved by using edge computing\footnote{In this paper, we refer to the term `edge computing' as computing in the `edge devices' such as smartphones and tablets, not the `edge server' in the edge of the internet (e.g., servers installed in 5G base stations).}, an alternative to cloud computing where data is processed physically close to their source. Edge devices are typically mobile with greater computational resources than dashcams, of which examples include smartphones, tablets, and laptops. Edge devices have a few advantages for real-time dashcam video analytics. First, they are widely used these days and we can expect one or more devices to be present in a driving vehicle, including the driver's smartphone. Second, recent edge devices are equipped with powerful resources but do not fully utilize them most of the time. The remaining resource capacity can be effectively used for video analytics. Third, most vehicles offer one or more charging ports, providing enough energy to edge devices for video analytics without draining their batteries. Finally, most edge devices have built-in cameras, allowing them to act as network-enabled dashcams for situations where a commercial dashcam is unavailable.

However, realizing real-time edge-based dashcam video analytics is not easy due to the following technical challenges:
\begin{itemize}
    \item \textbf{Heavy workload:} Real-time video analytics is computationally intensive, also accompanying a large amount of video data transfer from dashcams to edge devices. A single, average device may not be capable of analyzing video frames on time, necessitating the collaboration of multiple devices in a vehicle.
    \item \textbf{Unpredictable device connectivity:} The owner of each edge device has the full privilege to control the device, hindering the maintenance of a stable resource pool for video analytics. Existing devices may drop out or new devices may offer to contribute resources at any time.
    \item \textbf{Inter- and intra-device performance heterogeneity:} Edge devices may have different computing capacities from each other. Additionally, the computing capacity of an individual device may vary over time depending on the amount of owner-directed jobs and on physical characteristics such as device temperature.
    \item \textbf{Inter-video workload heterogeneity:} When using multiple dashcams, each video stream may require different AI-based analysis models that generate dissimilar workload amounts from each other.
\end{itemize}

In this paper, we present \name{}, a Distributed Edge-based dashcam Video Analytics system, that enables real-time dashcam video analytics using in-vehicle edge devices. The system analyzes video frames from two cameras (outward-facing and inward-facing) to help with safe driving\footnote{Although the current design of \name{} assumes two cameras, it can be easily extended to accommodate more than two cameras. Further, the usage of \name{} can be extended to other scenarios, such as military operations as well as rescue operations, with helmet-mount cameras.}. Specifically, video frames from the outward-facing camera (outer cam) are analyzed to detect potential driving hazards, and those from the inward-facing camera (inner cam) that captures the driver's face and hands are processed to identify driver distractedness. The whole process is performed fast enough to make prompt decisions based on the analysis results, increasing driving safety. The preliminary study of \name{} can be found in~\cite{edasha}.

\name{} implements several techniques to ensure low-latency and high throughput analytics, addressing the technical challenges above. First, it harnesses multiple edge devices in a vehicle to provide sufficient resources for compute-intensive video analytics tasks. It extracts video frames from two video sources and distributes them across consolidated edge devices to concurrently process them for higher throughput. Second, to cope with the changing resource availability due to the unpredictable device connectivity and intra-device performance heterogeneity, \name{} persistently monitors the connectivity and real-time performance of each device, and dynamically adjusts the dashcams' frame transfer rate to control the overall workloads considering the currently available computing capacity. Third, to overcome the inter-device performance heterogeneity and inter-video workload heterogeneity, \name{} implements a frame scheduler that distributes video frames taking both the device performance and video source into consideration. Finally, \name{} divides the entire execution flow into multiple independent phases and executes in a pipelined manner to increase the overall frame processing throughput.

We have implemented \name{} as an Android app and used heterogeneous smartphones to evaluate its performance in various environments assuming realistic scenarios. In addition, we have implemented another Android app that emulates a dashcam equipped with required functionalities for \name{} (i.e., communicating with an edge device and controlling its frame transfer rate), as such features are not present in commercial dashcams yet. Experimental results show that \name{} enables near real-time video analysis using commercial edge devices.
The contributions of this paper are threefold:
\begin{itemize}
    \item \textbf{Value addition to unused video data:} We propose an effective use case of dashcam video data that is otherwise discarded without any use.
    \item \textbf{Enabling techniques for edge-based video analytics:} We design various techniques for low-latency, high-throughput real-time video analytics using unused resources in one or more nearby mobile devices.
    \item \textbf{Application-level implementation:} We implement the proposed system in mobile applications without modifying the Android platform. In addition, our dashcam emulation application enables us to provide video analytics services for vehicles that lack actual dashcams.
\end{itemize}

The rest of this paper is organized as follows. In section~\ref{sec:related_work}, we present prior efforts related to this work. We design the proposed system in Sections~\ref{sec:design}. In Section~\ref{sec:evalutation}, we evaluate the performance of the proposed system under various environments. We then present a few lessons learned while evaluating the proposed system in Section~\ref{sec:lesson}. Finally, we conclude our work in Section~\ref{sec:conclusion}.

\section{Related Work}\label{sec:related_work}
Until a decade ago, the main research avenue to provide computationally heavy data processing services using data generated at the edge of the network was to deploy the services on the cloud, which accompanies a high and unstable data transfer latency due to the long distance between data source and cloud data centers. Recently, as the amount of computational resources in mobile/IoT devices is ever-increasing, many attempts have been made to use them for data processing to avoid high latency~\cite{mobilePhoneComputing,rejeb2022big}. One such effort is real-time video analytics using mobile/IoT devices.

In \cite{huynhDeepMonMobileGPUbased2017}, they proposed various optimization techniques for offloading convolutional layers of a deep learning model to mobile GPU to accelerate image processing in mobile devices. In \cite{10.1145/3466772.3467037}, they proposed a reliable motion extraction scheme from video frames to expedite object tracking in mobile devices. In \cite{sunDetectingCounterfeitLiquid2022}, they presented an application of image analytics using smartphones, which detects adulterated liquid products without opening the bottles. A resource-efficient video analytics model based on a spatial attention mechanism and multidomain networks has been proposed in~\cite{visualObjectDetectionAndTracking} to detect and track visual objects in resource-constrained IoT devices. While the above works showed the feasibility of using mobile/IoT devices for image analytics, they used a single device without considering device consolidation for higher throughput and lower latency.

A more recent work, CrossVision \cite{10202594}, considered video analytics using multiple smart cameras in proximity. It reduces the image analysis latency by removing the RoI (Region-of-Interest) redundancy in overlapped FoVs (Field-of-Views) among proximate cameras. While CrossVision does not migrate tasks from one device (smart camera) to another, it achieves workload balancing by imposing the analysis task of redundant RoI to a single light-loaded device among those sharing the RoI. However, it is not effective in videos without overlapped FoVs like in \name{}.

To overcome the resource limitation in a single mobile device, some prior works have managed to make a local network of nearby devices to provide a larger resource pool for heavy tasks in mobile devices. Cloudlet exploits a nearby server to provide VM-based customized services to mobile users~\cite{cloudlet}. FemtoClouds leverages mobile devices to provide cloud services at the network edge~\cite{habakkarimFemtoCloudsLeveraging2015}. In \cite{dangeSchedulingTaskCollaborative2016} and \cite{fernandoMobileCrowdComputing2012}, they proposed rank-based and work stealing-based task scheduling schemes, respectively, for consolidated mobile devices. The above efforts provide architectures and scheduling schemes in general, for consolidated mobile devices, without considering the unique characteristics of video analytics applications.

Other research areas for networked mobile devices include multi-mobile computing and mobile collaboration. Prior efforts in multi-mobile computing aim to use resources (hardware, software, data, and application) in nearby devices to enhance user experiences beyond what a single device can offer \cite{rio,mobileplus,collaboroid,flux,samd,alduaijHeterogeneousMultiMobileComputing2019}. They use Device-to-Device (D2D) communications to establish connections among nearby mobile devices~\cite{googleNearbyConnections2021,jenschke2023nearby}. Mobile collaboration enables mobile users to achieve a shared goal by collaborating with each other. For example, SignalGuru attempts to reduce car fuel consumption by predicting traffic light changes through collaborative analysis of traffic lights captured by dash-mounted smartphones in multiple cars~\cite{koukoumidisSignalGuruLeveragingMobile2011}.

For the specific purpose of improving driving safety, various video analytics models have been proposed to detect road hazards including other vehicles \cite{chadwickDistantVehicleDetection2019,rybskiVisualClassificationCoarse2010,tehraniniknejadOnRoadMultivehicleTracking2012,zhangRealtimeVehicleDetection2018}, pedestrians  \cite{tomeDeepConvolutionalNeural2016,dominguez-sanchezPedestrianMovementDirection2017,liuRobustFastPedestrian2013}, and general road obstructions \cite{creusotRealtimeSmallObstacle2015,moralesrosalesOnroadObstacleDetection2018,jiaRealtimeObstacleDetection2015}. Video analytics models to detect driver distractedness have also been proposed \cite{yanDriverBehaviorRecognition2016,tranRealtimeDetectionDistracted2018,kapoorRealTimeDriverDistraction2020,huangHCFHybridCNN2020}. Some prior works proposed video analytics systems that detect both driving hazards and driver distractedness \cite{jainCarThatKnows2015,rezaeiLookDriverLook2014,dey2021context}, which is the main goal of \name{}.
Such video analytics models have different computational complexities each other, affecting the frame processing throughput. For example, a pedestrian detection model proposed in \cite{tomeDeepConvolutionalNeural2016} achieved a frame rate of 21.7 FPS, and the driver distractedness detection model proposed in \cite{huangHCFHybridCNN2020} showed a frame rate of 24 FPS, both on a desktop PC. The above works are orthogonal to this work, which focuses on optimizing latencies in the system other than the video analysis latency by orchestrating consolidated mobile devices. Machine-learning libraries for mobile devices such as TensorFlow Lite~\cite{googleTensorFlowLiteML2021} enable us to use any compatible video analytics model in \name{}.

\section{\name{} Design}\label{sec:design}

\name{} system largely consists of two components, \textit{coordinator} and \textit{worker}. It assumes two types of edge devices, one primary, and zero or more external devices, depending on which component of \name{} is active in the device. The coordinator is active only in the primary device, while the worker is active in both primary and external devices. Hence, \name{} can work even when there is only one edge device in a vehicle as it will act as the primary device.

\begin{figure}[t]
    \centering
    \includegraphics[width=0.9\columnwidth]{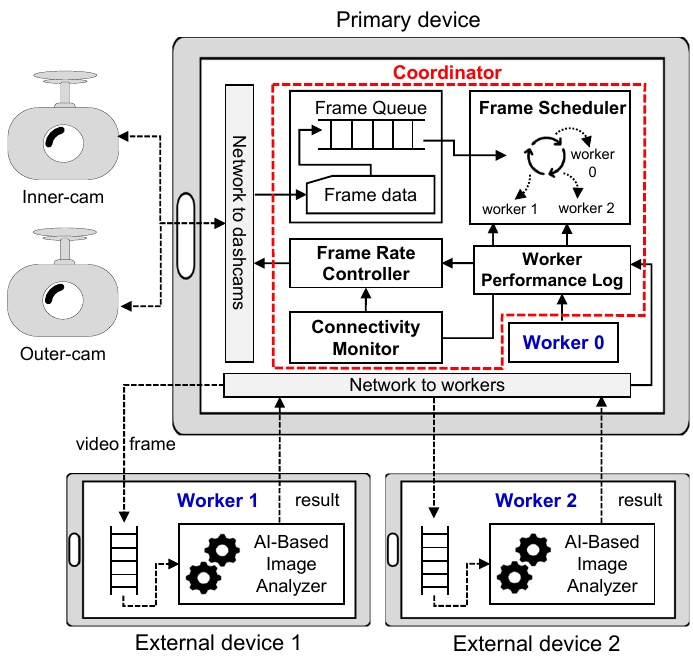}
    \caption{Overall architecture of \name{}.}
    \label{fig:new_overview_0612}
\end{figure}

Figure~\ref{fig:new_overview_0612} shows the overall architecture of \name{}. Dashcams, primary and external devices form an internal network using Wi\nobreakdash-Fi Direct~\cite{WiFi-direct} or mobile AP (hotspot). The coordinator in the primary device is responsible for constructing and maintaining the system. It also coordinates the overall video analysis process among workers in all devices. Workers in both the primary and external devices analyze video frames and send the results to the coordinator. We call the worker running in the primary device as worker 0, and other workers in external devices as worker $n \;(n > 0)$.

Video frames captured by dashcams are transferred to the coordinator and added to the frame queue. The frame scheduler fetches video frames from the queue one at a time and transfers them to workers for analysis. To that end, it maintains the sequence of workers that will analyze upcoming frames and transfers each video frame to the workers as per the sequence. Determining the worker sequence that balances loads across heterogeneous workers is crucial to minimize the overall processing latency. A video frame transferred to a worker is enqueued to the worker's frame queue until the image analyzer processes it. The analysis result is sent back to the coordinator, along with the worker's measured performance index. The coordinator determines whether or not to generate an alarm signal based on the result, e.g., a pothole on the road has been detected, or the driver is using a phone. The measured performance index is appended to the worker performance log in the coordinator to be used by both the frame scheduler and frame rate controller. The frame scheduler uses the log to construct the worker sequence. The frame rate controller estimates overall system capacity by analyzing the log and monitoring device connectivity. It dynamically adjusts the frame rate to maintain workload within the estimated capacity. The modified frame rate is then sent to dashcams which instructs them to transfer video frames at this rate.

In this section, we detail the core techniques of \name{} which address technical challenges shown in Section~\ref{sec:introduction}. In brief, task pipelining among consolidated devices enables \name{} to achieve high frame processing throughput despite its heavy workloads (Section \ref{sec:pipeline}). The weighted frame scheduling enables \name{} to overcome the inter-device performance heterogeneity and inter-video workload heterogeneity (Section \ref{sec:scheduling}). The intra-device performance heterogeneity and unpredictable device connectivity issues are resolved by dynamic frame rate control (Section \ref{sec:frame-rate-control}).

\subsection{Pipelined Task Execution}\label{sec:pipeline}

In addition to being computationally intensive, real-time video analytics on consolidated edge devices involves a large amount of video data transfers between devices which are also time-consuming. \name{} identifies three time-consuming tasks in the entire video analytics process: 1) video frame transfer from dashcams to the coordinator, 2) video frame transfer from the coordinator to workers, and 3) image analysis by workers. Since those tasks have a temporal order, 1) $\rightarrow$ 2) $\rightarrow$ 3), their serial execution leaves other tasks idle while a task is being executed, decreasing the resource utilization and frame processing throughput in modern edge devices that equip multi-core CPUs. \name{} resolves this problem by pipelining the tasks, as shown in Figure~\ref{fig:pipelining}. By executing all the tasks for consecutive frames in parallel using multiple cores, \name{} can fully utilize the device resources and improve the overall frame processing throughput.

\begin{figure}[t]
    \centering
    \includegraphics[width=0.85\columnwidth]{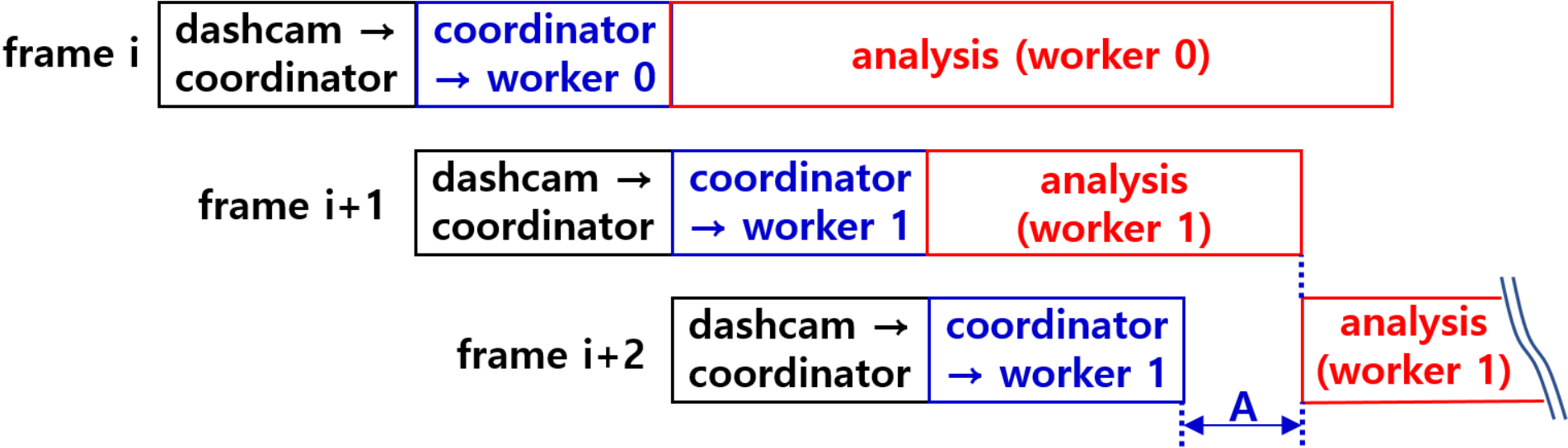}
    \caption{Task pipelining in \name{} assuming two devices. A: waiting time of the analysis task for frame $i+2$ until finishing the analysis of preceding frame $i+1$ in worker 1.}
    \label{fig:pipelining}
\end{figure}

The pipelining architecture of \name{} provides two meaningful guidelines to maximize the pipelining effect. First, a high level of concurrency is required in edge devices, especially in the primary device that acts as both the coordinator and worker. The primary device is involved in three kinds of data transfer tasks (dashcams$\rightarrow$coordinator and coordinator$\rightarrow$workers frame transfers, and results collection from workers), an image analysis task for worker 0, and various coordination tasks including frame scheduling and frame rate control, all of which can be executed concurrently. To optimize overall system performance, \name{} should ideally use the most powerful edge device available in the car as the primary device.
Second, a carefully designed frame scheduling policy that precisely considers different image analysis times in each device is indispensable to the coordinator. For example, in Figure~\ref{fig:pipelining}, if frame $i+2$ is assigned to worker 0 instead of worker 1, its processing can be delayed until finishing the analysis task for frame $i$ in worker 0. Since the delay is larger than that when the frame is assigned to worker 1 as the figure exemplifies, such frame assignment increases the overall processing latency. Hence, a simple round-robin scheduling policy is not adequate to \name{} which needs to consider inter- and intra-device performance heterogeneity among participating devices.

\subsection{Weighted Frame Scheduling}\label{sec:scheduling}

The image analysis task takes the largest time in the entire frame analysis pipeline of \name{}, and two factors determine the time; the workload complexity for image analysis and the available computing capacity in each device.

The workload complexity for image analysis is generally determined by the AI model used for the image analysis. In \name{}, for instance, the AI model analyzing video frames from the outer camera involves far more operations than the one processing frames from the inner cam. This disparity arises because the outer cam model detects multiple arbitrary objects in each frame to identify various driving hazards, whereas the inner cam model only detects some face parts (eyes and ears) and the hands to assess driver distractedness. The frame scheduler needs to distinguish between two video sources to overcome the inter-video workload heterogeneity. However, since all workers in \name{} process frames from both video sources, they have identical workload complexity.

The available computing capacity in each device is determined by the native computing power of the device and the amount of owner-directed or background jobs that the device is currently executing. The former depends on the device specifications leading to inter-device performance heterogeneity, while the latter leads to intra-device performance heterogeneity. Notably, the native computing power of a mobile device may even change over time depending on the device's temperature, which is another source of intra-device performance heterogeneity. For example, the thermal governors in Android drop CPU frequency when the device temperature exceeds the thermal limit~\cite{thermal}. The frame scheduler in the coordinator needs to carefully incorporate both performance heterogeneities to reduce the per-frame processing latency.

However, numerically estimating the available computational resources on each device at a specific moment is challenging. Instead, we adopt a black-box approach, focusing solely on monitoring the output: the \textit{analysis time} for a single video frame. Using the analysis times for the recent video frames in each worker, we can estimate the \textit{relative} computing capacity currently available among workers and distribute upcoming video frames in proportion to the relative capacity. Such relative performance-based frame distribution enables us to focus not on the complex estimation of the time-varying computing capacity in each device but on distributing video frames across workers with different computing capacities which \name{} numerically represents by \textit{worker weights}.

\begin{figure}[t]
    \centering
    \includegraphics[width=0.75\columnwidth]{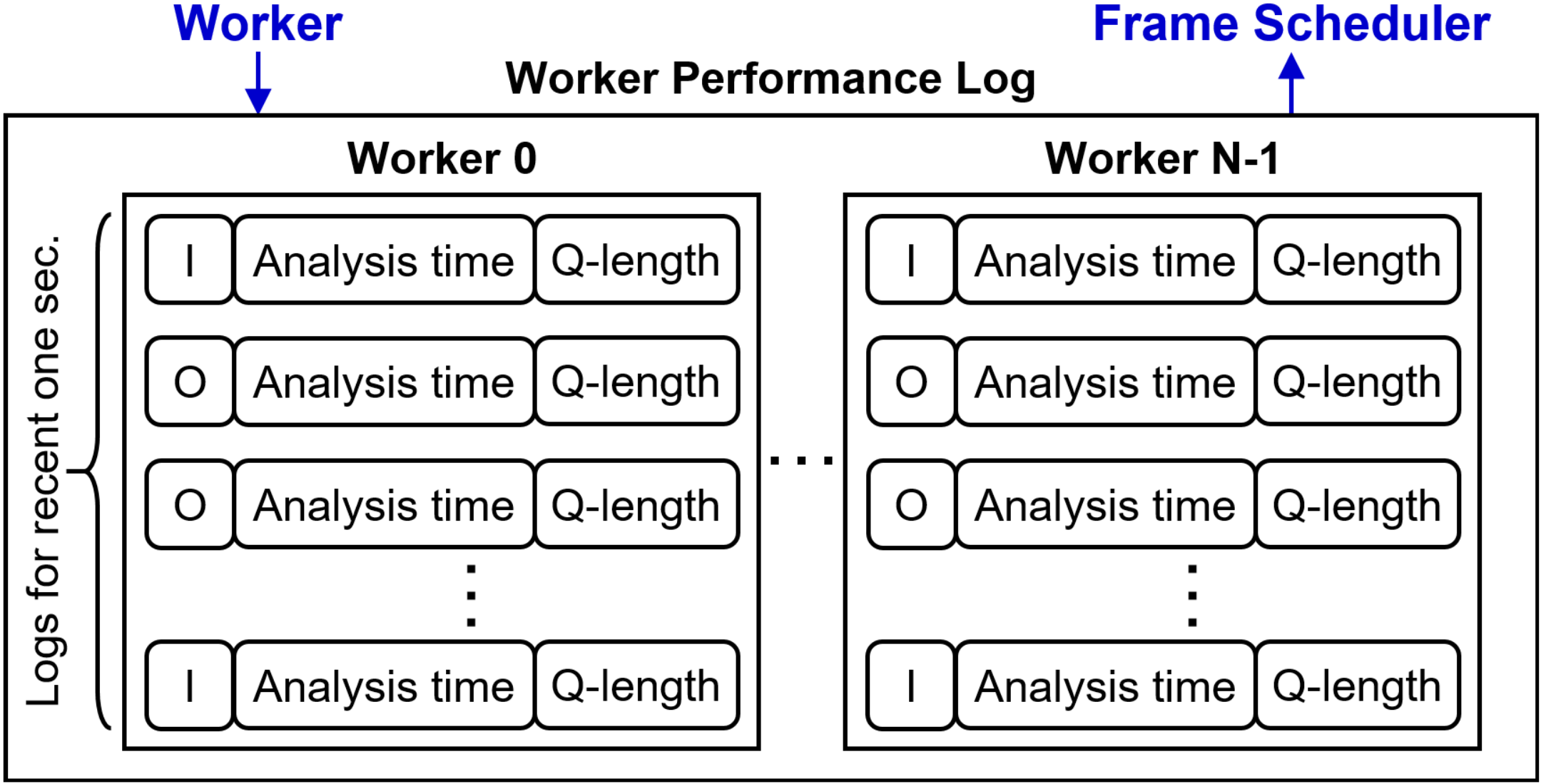}
    \caption{Worker performance log.}
    \label{fig:performance-log}
\end{figure}

To measure the time-varying worker weight, \name{} maintains the \textit{worker performance log} in the coordinator as shown in Figure~\ref{fig:performance-log}. The worker performance log contains the performance index records for each frame analysis task, grouped by the worker. Each record is a triple of (video source, analysis time, queue length), in which the video source is either inner (I) or outer (O) camera and the queue length denotes the number of waiting frames in the worker's frame queue after analyzing the frame. The analysis time is the consumed time by the image analyzer to process the frame. A record is generated by a worker and appended to the log by the coordinator upon receiving an image analysis result from a worker, and the coordinator maintains only recent records (within one second) for each worker to incorporate the recent computing capacities of workers to worker weights.

\subsubsection{Worker Weight Calculation}\label{sec:weight}

The weights of workers represent the currently available computing capacities in each worker and are calculated using the worker performance log. To numerically represent the computing capacity of a worker, we estimate the internal processing latency for an upcoming frame in the worker and then take its inverse to get the weight of the worker. The performance index records provide sufficient information to estimate the internal frame processing latency. Specifically, the queue length indicates the queueing delay until analyzing the upcoming frame and the analysis time represents the computing latency for analyzing the frame. The frame scheduler calculates average analysis times and queue lengths for each worker based on all recorded data. It differentiates between video sources when determining average analysis times.

The frame scheduler first estimates the computing latency for a single frame analysis in worker $i$, $T_C^i$, as $T_C^i=(T_I^i+T_O^i)/2$, where $T_I^i$ and $T_O^i$ denote the average analysis times for frames from inner and outer cameras, respectively. By averaging the \textit{average} frame analysis times from different video sources, $T_C^i$ overcomes the possible under- or over-estimation of the computing latency when the numbers of performance index records for each video source are not identical, provided that there is inter-video workload heterogeneity. For a newly joined device that has no performance index record in the worker performance log, default values for $T_I^i$ and $T_O^i$ are used. Then, the frame scheduler estimates the internal processing latency for an upcoming frame by multiplying $T_C^i$ by the average queue length of the worker, $L_Q^i$, plus the upcoming frame itself, i.e., ($L_Q^i$ + 1). Finally, by taking its inverse, we get the weight of worker $i$, $w_i$, as;
$$w_i = \frac{1}{T_C^i(L_Q^i+1)}$$

\subsubsection{Worker Sequence Generation}\label{sec:seq_gen}

\begin{figure}[t]
    \centering
    \includegraphics[width=0.7\columnwidth]{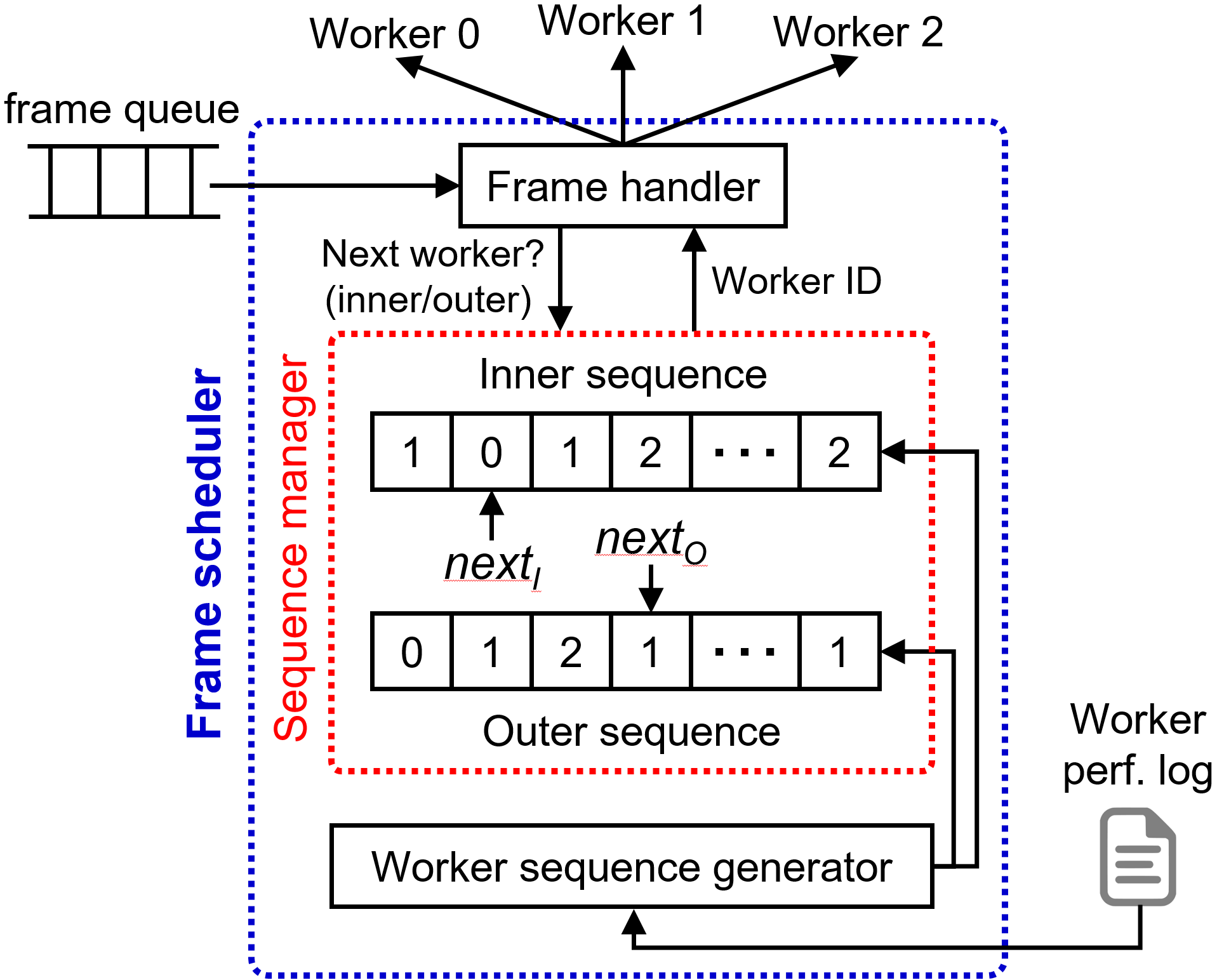}
    \caption{Worker sequence-based frame scheduling.}
    \label{fig:scheduler}
\end{figure}

The frame scheduler in \name{} distributes video frames to workers as per pre-determined fixed-sized worker sequences, instead of determining workers on a frame-by-frame basis, to reduce the scheduling overhead in the coordinator. As shown in Figure~\ref{fig:scheduler}, it maintains distinct worker sequences for each video source, inner and outer sequences, representing the order of workers that will process upcoming video frames from inner and outer dashcams, respectively. The separated worker sequences per video source prevent the biased assignment of frames from a particular source to a specific worker. Upon fetching a frame from the frame queue, the frame handler identifies its source and requests a worker ID that will process the frame to the sequence manager. The sequence manager keeps pointers ($next_I$ and $next_O$) to the next workers in each sequence. It replies with the worker ID indicated by the pointer and moves the pointer to the next worker ID. Then the frame handler transfers the frame to the designated worker.

The worker sequences are generated in two cases. First, when a worker pointer for a sequence reaches the end of the sequence, the sequence manager triggers the worker sequence generator to establish a new sequence. Second, when the coordinator detects any change (join or leave) in the devices' connection states, it immediately triggers the worker sequence generator to rebuild both sequences entirely, incorporating the changed device set. The length of the worker sequences is short enough to be consumed quickly so that sudden changes to the worker weights could be immediately applied through timely sequence generations. The core of the frame scheduler is its worker sequence generation algorithm.

A worker sequence is comprised of a fixed number of slots, each containing a worker ID. The worker sequence generation algorithm fills up the slots aiming to achieve its design goal; \textit{minimizing the overall frame processing latency}. The algorithm achieves the goal using two strategies.

\textbf{Load-balanced Frame Distribution.} The algorithm distributes upcoming frames to workers in proportion to workers' computing capacity, represented by worker weights, to avoid latency caused by overburdening a worker.

\textbf{Even Intra-worker Frame Distribution.} The algorithm keeps the gaps between the same worker in the sequence as even as possible. A burst transfer of frames to a worker instantaneously increases the length of the worker's frame queue, raising the processing latency of later enqueued frames. By evenly distributing the same worker in a sequence, we can provide the worker as much time as possible to process the preceding frame before transferring a subsequent frame, decreasing the overall queueing delay.

\begin{algorithm}[h]
\small
    \SetKwComment{Comment}{\scriptsize // }{}
    \SetCommentSty{scriptsize}
    \SetKwInput{N}{$N$}
    \SetKwInput{M}{$M$}
    \SetKwInput{pi}{$p_i$}
    \SetKwInput{wi}{$w_i$}
    \SetKwInput{Wall}{$W$}
    \N{the length of the worker sequence}
    \M{the number of workers}
    \pi{priority value of worker $i$ (initially, zero)}
    \wi{weight of worker $i$}
    \Wall{sum of worker weights for all workers ($=\sum_{i=0}^{M-1} w_i$)}
    \BlankLine
    \KwResult{a worker sequence ${\cal S}[N]$ (initially, empty)}
    \BlankLine
    \Begin{
        \For{$0 \leq pos \leq N-1$}{
            $p_{max} \gets 0$\Comment*[r]{$p_{max}$: highest priority value ever seen}
            \For{$0 \leq i \leq M-1$}{
                $p_i \gets p_i + w_i$\Comment*[r]{increase all priority values by weights}
                \If{$p_i > p_{max}$}{
                    $p_{max} \gets p_i$\;
                    $next \gets i$\Comment*[r]{$next$: worker with the highest priority}
                }
            }

            ${\cal S}[pos] \gets next$\Comment*[r]{append worker $next$ to ${\cal S}$}
            $p_{next} \gets p_{next}-W$\Comment*[r]{decrease worker $next$'s priority value}
        }
    }
    \caption{Worker sequence generation algorithm}
    \label{alg:algo_sequence}
\end{algorithm}

Algorithm~\ref{alg:algo_sequence} shows the procedure for creating a worker sequence among $M$ workers for the upcoming $N$ frames, which can be used for generating both the inner and outer sequences. While it seems, at a glance, like a weight-based round-robin scheduling algorithm, it generates a different sequence from that of the round-robin algorithm due to its proportional increase and fixed decrease policy for priority values. The priority values of workers increase by their respective weights ($w_i$) in every iteration of the outer loop, while they decrease by the sum of weights ($W$) when appended to the sequence, always maintaining the sum of priority values to be zero at the end of each iteration of the outer loop. This policy makes the number of slots each worker gets in a sequence to be proportional to the worker's weight, as shown in Theorem~\ref{theorem1}, achieving the load-balanced frame distribution.

\begin{theorem}\label{theorem1}
The worker sequence generation algorithm assigns slots in a sequence in proportion to the worker weight.
\end{theorem}

\begin{proof}
Let $s_i$ be the number of slots assigned to worker $i$ upon finishing the algorithm. During the algorithm execution, the priority value of worker $i$ increases $N$ times each by $w_i$ and decreases $s_i$ times each by $W$. Hence, after finishing the algorithm, $p_i = N\cdot w_i - W\cdot s_i$. Redistributing terms in the equation, we get $$s_i=\frac{1}{W}(N\cdot w_i - p_i).$$ Since $-W < p_i \leq W$ and $s_i$ is an integer,
$$\lfloor N\cdot \frac{w_i}{W} -1\rfloor \leq s_i < \lceil N\cdot \frac{w_i}{W} +1\rceil .$$
Hence, $s_i$ is proportional to $w_i$.
\end{proof}

Another powerful effect of this algorithm is that it prevents the same worker from getting multiple slots in a row and then not getting a slot for a while, achieving even intra-worker frame distribution, owing to its intrinsic round-robin-based slot distribution policy. This can be easily seen by assuming workers have similar weights. In that case, workers repeatedly get a slot in their weight order, keeping an identical distance of $M$ between the same worker in the sequence.

\begin{table}[h]
\setlength{\tabcolsep}{4pt}
\setlength{\belowrulesep}{0ex}
\setlength{\aboverulesep}{0ex}
\centering
\small
\caption{An example of the worker sequence generation algorithm when $N = 9$, $M = 3$, $w_0 = 2$, $w_1 = 3$, $w_2 = 4$.}
\begin{tabular}{@{}c|c|ccccccccc|c}
\toprule
&\multirow{2}{*}{$w_i$}&\multicolumn{9}{c|}{$p_i$ after each iteration (outer loop)}& \multirow{2}{*}{\# slots}\\ \cmidrule{3-11}
& & 1 & 2 & 3 & 4 & 5 & 6 & 7 & 8 & 9 &  \\ \midrule
worker 0 & 2 & 2 & 4 & \circled{-3} & -1 & 1 & 3 & \circled{-4} & -2 & 0 & 2  \\
worker 1 & 3 & 3 & \circled{-3} & 0 & 3 & \circled{-3} & 0 & 3 & \circled{-3} & 0 & 3 \\
worker 2 & 4 & \circled{-5} & -1 & 3 & \circled{-2} & 2 & \circled{-3} & 1 & 5 & \circledzero{0} & 4  \\ \midrule
\multicolumn{2}{c|}{Sequence} & 2 & 1 & 0 & 2 & 1 & 2 & 0 & 1 & 2 & 9     \\ \bottomrule
\end{tabular}
\label{tab:worker-sequence-example}
\end{table}

Table~\ref{tab:worker-sequence-example} shows an execution example of the algorithm. The weight values are assumed to be integers for simplicity because only their relative ratio matters, but in reality they can be any positive real value. In each iteration, worker $i$ gains $w_i$ of its priority value, and a worker with the highest priority value (worker $next$) is chosen to be appended to the worker sequence. Then, the priority value of the chosen worker is decreased by 9 (=$W$) and is marked as circled in the table.
We can see that the algorithm achieves both load-balanced frame distribution and even intra-worker frame distribution. With a larger $N$ than 9, the above worker sequence repeats, since priority values in the 9th iteration are identical to their initial values (i.e., all 0s).

The separate worker sequences for inner and outer videos enable \name{} to overcome the inter-video workload heterogeneity by evenly distributing inner and outer frames across workers. Without such distinction, a worker may receive many frames from one source that requires complex computation for analysis, resulting in higher overall latency, while another worker receives many frames from the other source that requires less computation, leaving a portion of its computing capacity unused. For example, assuming a unified worker sequence and two workers with similar computing capacity, the interleaved arrival of frames from two sources may result in an entirely separated distribution of frames per their sources.

\subsection{Dynamic Frame Rate Control}\label{sec:frame-rate-control}

Due to the unpredictable device connectivity and the intra-device performance heterogeneity, the overall computing capacity in \name{} varies over time.
With low computing capacity, \name{} cannot process video frames promptly, leading to their accumulation in frame queues. This increases the overall latency and results in outdated outputs, such as detecting a road hazard too late to avoid. In that case, we need to reduce the burden on workers by decreasing the frame rate in dashcams so that workers can process frames without excessive delays. On the other hand, with high computing capacity, we need to increase the frame rate to reduce the inter-frame time gap for faster detection of driving hazards or driver distractedness.

\name{} implements a dynamic frame rate controller that detects the changes in the overall computing capacity and updates the frame rate based on the queueing system model. As shown in Figure~\ref{fig:frame-rate-control}, it uses the performance index records in the worker performance log to estimate the changing computing capacity of each connected device. The connectivity monitor informs the frame rate controller of the join or leave event of a new or existing device. The updated frame rate ($r_{new}$) by the frame rate controller is informed to dashcams so that they transfer frames to the coordinator at the updated rate\footnote{We assume that dashcams have the functionalities to communicate with the coordinator and control their frame transfer rate for this purpose. Until such dashcams become available in the market, smartphones that install our dashcam emulation app can replace them.}. Notably, dashcams still maintain their native frame rate ($F_R$) for capturing and storing videos, to use them as unimpaired evidence for traffic accidents. They change only the frame \textit{transfer} rate by omitting a portion of frames in the recorded videos to transfer. Before detailing the behavior of the dynamic frame rate controller, it is noteworthy that the frame rate control aims to control the total amount of load on the entire system, while the individual load control on each device is the frame scheduler's role.

\begin{figure}[t]
    \centering
    \includegraphics[width=0.75\columnwidth]{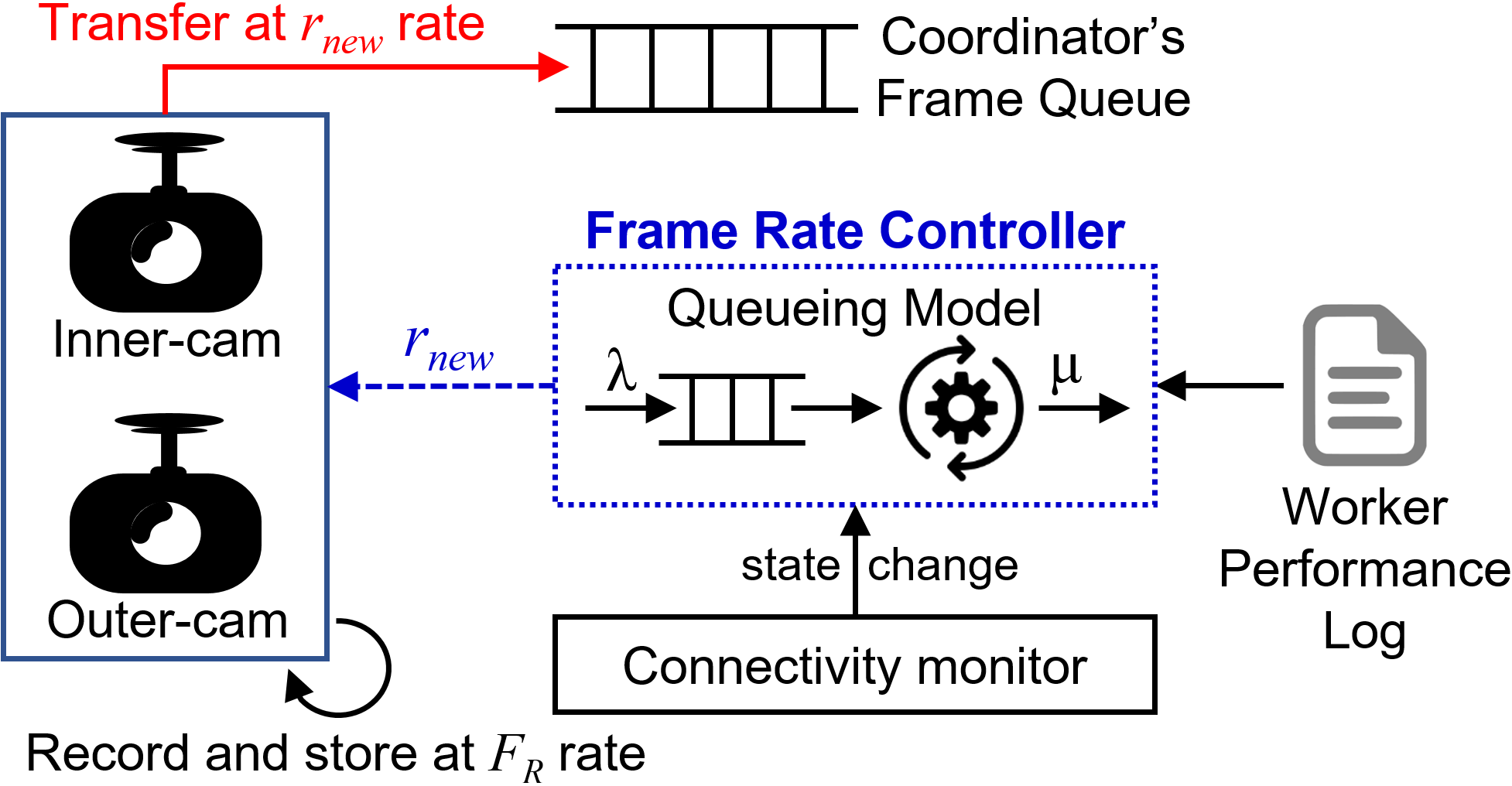}
    \caption{Dynamic frame rate control.}
    \label{fig:frame-rate-control}
\end{figure}

\name{} periodically adjusts the frame rate under the changing computing capacity by invoking the frame rate controller. The controller estimates the average frame processing time (= queueing delay + image analysis time) in a worker and uses a latency deadline ($L_D$) which is a configurable parameter representing the upper bound of the frame processing latency for detecting driving hazards or driver distractedness in a frame. To estimate the average frame processing time, it employs the $M$/$M$/1 queueing system model~\cite{queueing} assuming that frame inter-arrival times and image analysis times in a worker follow exponential distributions. The average frame processing time in the model is $1/(\mu - \lambda)$, where $\mu$ and $\lambda$ denote the average frame analysis throughput and frame arrival rate, respectively.

Given the average image analysis time in worker $i$, $T_C^i$, the frame analysis throughput in worker $i$ is $\mu_i=d_p/T_C^i$, where $d_p$ is the degree of parallelism representing the number of frames concurrently processed in a multicore system. In a multicore system, as in modern smartphones, we can benefit from parallel frame analysis by using more than one frame analysis thread. The TensorFlow Lite inference engine~\cite{googleTensorFlowLiteTask} used in \name{} enables users to manually set the number of concurrent threads for its inference engine. Hence, in an eight-core device, for example, we can configure TensorFlow to process two frames in parallel ($d_p = 2$), each using four cores\footnote{Although $M$/$M$/$m$ queueing system model can be more accurate for multicore systems, we use $M$/$M$/1 model (with the throughput multiplied by the degree of parallelism) to exploit its computational simplicity.}. Then we can model the average frame processing time in worker $i$ as
$$T_{avg}^i=\frac{1}{d_p/T_C^i - \lambda_i},$$
where $\lambda_i$ denotes the average frame arrival rate to worker $i$. The frame rate controller finds $\lambda_i$ that satisfies
$$T_{avg}^i + 2T_F + T_R\leq L_D ,$$ where $T_F$ and $T_R$ denote the average frame and result transfer times, respectively, each is modeled as a constant value by dividing the average frame size or result size by the network bandwidth. We double the frame transfer time to incorporate two frame transfers from a dashcam to the coordinator and from the coordinator to a worker. Solving the inequality, we get the upper bound of the frame arrival rate for worker $i$ that satisfies the latency deadline as;
$$\lambda_i \leq \frac{d_p}{T_C^i} - \frac{1}{L_D - 2T_F - T_R}$$
Finally, the frame rate controller determines the new frame rate by summing up the upper bounds for all $M$ workers;
$$r_{new} = \min\left(\sum_{i=1}^M ( \frac{d_p}{T_C^i} - \frac{1}{L_D - 2T_F - T_R}), F_R\right)$$

In the above equation, all other parameters than $T_C^i$ and $M$ are constants. $T_C^i$, the average image analysis time in worker $i$, changes depending on the varying computing capacity of the worker. When a user launches a new app or stops a running app on a device, the computing capacity change accompanied by the action is reflected to $T_C^i$. Note that, differently from the definition of $T_C^i$ in Section~\ref{sec:weight}, we use $T_C^i$ in the frame rate controller as the larger average analysis time between inner and outer frames, i.e., $T_C^i = \max(T_I^i, T_O^i)$, instead of averaging them, so that we can satisfy the latency deadline for both inner and outer frames.

The number of workers ($M$) changes when the device connectivity state changes. Upon receiving a device join event from the connectivity monitor, the frame rate controller uses a pessimistic approach to estimate the computing capacity (represented by $T_C^i$) of the new device whose performance log records are missing. Specifically, it uses the largest $T_C^i$ value from the other devices, aiming to underestimate the device capacity for the sake of safety. More accurate computing capacity will be estimated shortly as the performance index records of the new device are accumulated. When a device-leave event occurs, the computing capacity of the leaving device is naturally excluded from the system capacity by the above formula with a decreased $M$.

\section{Performance Evaluation}\label{sec:evalutation}

In this section, we first evaluate the performance of \name{} in terms of frame processing throughput and latency under various realistic scenarios. Then we investigate the internal behaviors of \name{} in detail.

\subsection{Application Implementation}
We have implemented \name{} as a single Android app\footnote{The source code of \name{} will be publicly available upon acceptance.},
rather than individually implementing the coordinator and worker apps to avoid costly inter-process communication between them in the primary device.
Users can configure their devices to act as either primary or external devices. If a device is configured to act as the primary device, both the coordinator and worker functionalities are enabled while only the worker is enabled in an external device.

\name{} requires an immediate frame transfer to the coordinator upon frame capture in a dashcam to minimize the overall analysis latency. However, commercial dashcams generate a video \textit{file} employing inter-frame compression, which prohibits instant frame transfer. To resolve the problem, we implemented our dashcam emulation app which generates individual frame images (bitmap format) instead of a video file. Each bitmap image is then independently compressed and transferred to the coordinator. Furthermore, since data serialization is a prerequisite for communication between Android devices, we used JPEG compression which generates output in a serializable form. Android APIs support JPEG compression of bitmap objects into byte arrays and vice versa. Decompressed bitmap images in the worker are fed into the image analyzer (TensorFlow Lite) for analysis.

However, due to JPEG's lossy compression algorithm, the accuracy of image analysis results may decrease.
As a result, JPEG compression introduces a trade-off between the frame transfer latency and the analysis accuracy, necessitating an acceptable compromise between them. To that end, we measured the average accuracy and compressed image size for various JPEG quality levels for all images in the video file we used in our experiments.

\begin{table}[h]
\centering
\scriptsize
\setlength{\tabcolsep}{6.5pt}
\caption{JPEG compression ratio versus analysis accuracy.}\label{tab:tradeoff}
\begin{tabular}{@{}llrrrrrr@{}}
\toprule
Quality level&&100&90&80&70&60&50\\ \midrule
\multirow{2}{*}{Compressed size}&Outer&0.177&0.061&0.042&0.033&0.028&0.024\\
&Inner&0.157&0.055&0.037&0.029&0.024&0.020\\
\midrule
\multirow{2}{*}{Accuracy (\%)}&Outer&89.0&88.0&87.0&86.0&85.6&85.0\\
&Inner&98.8&98.1&98.0&97.4&96.8&96.9\\
\bottomrule
\end{tabular}
\end{table}

Table~\ref{tab:tradeoff} shows the results. The compressed image size represents the normalized size of the compressed image relative to the original image size, and the accuracy is also normalized so that the accuracy with the original image is 100\%. Lowering the JPEG quality level from 100 to 80 greatly decreases image size without causing a noticeable decline in accuracy. Hence, we set the quality level to 80 in all our experiments.

\subsection{Experimental Environment}

We conducted experiments in an indoor environment, instead of the real driving environment, to provide identical driving situations to every experimental configuration. To that end, we used smartphones with our dashcam emulation app installed. The app implements the required functionalities for \name{}, including communication protocol with the coordinator, image compression and transfer, and frame rate control. It accepts video data either from a device's camera via the Android API module \texttt{android.hardware.camera2}, or from a video file specified by the user. The input data for our experiments are outer and inner dashcam video files, shown in Table~\ref{tab:video_list}, recorded for research purposes.

\setlength{\tabcolsep}{4pt} 
\begin{table}[h]
\centering
\scriptsize
\caption{Video files used in experiments.}
        \begin{tabular}{@{}llccc@{}}
        \toprule
        Source &Dataset (file name)&Resolution&FPS&Length\\
        \midrule
        Inner cam & DMD driver monitoring dataset\cite{DMD-dataset} (file A\tablefootnote{File A: ``gA\_1\_s3\_2019-03-14T14;31;08+01;00\_rgb\_face.mp4". We have transcoded this video into H.264/MPEG-4 AVC compression format that our dashcam emulation app can decode, using \texttt{x264}~\cite{x264} library.})&1280x720&30&61 sec.\\
        Outer cam & BDD100K driving video dataset\cite{bdd100k} (file B\tablefootnote{File B: ``b1c66a42-6f7d68ca.mov"})&1280x720&30&40 sec.\\
        \bottomrule
        \end{tabular}
    \label{tab:video_list}
\end{table}

To exclude the I/O delay in reading video files, which would not happen in a real dashcam, the dashcam emulation app loads the entire video file into memory before starting an experiment. When the video length is shorter than the experiment time, the app loops over the video frames stored in memory.
It controls the frame (transfer) rate by dropping a portion of frames from the input video files. The average frame sizes, after JPEG compression, of inner and outer videos are 101 KB and 116 KB, respectively.

\setlength{\tabcolsep}{4pt} 
\begin{table}[h]
\centering
\caption{List of devices used in experiments.}
\scriptsize
        \begin{tabular}{@{}lccll@{}}
        \toprule
        Device & Category & \# devices & CPU cores & Role \\
        \midrule
        Find X2 & Strong & 1 & 8 (1x2.8, 3x2.4, 4x1.8 GHz) & Primary/External \\
        Oneplus 8 & Strong & 1 & 8 (1x2.8, 3x2.4, 4x1.8 GHz) & Primary/External \\
        Galaxy S22 & Strong & 1 & 8 (1x3.0, 3x2.5, 4x1.8 GHz) & Primary/External \\
        Pixel 5 & Weak & 3 & 8 (1x2.4, 1x2.2, 6x1.8 GHz) & Primary/External \\
        Pixel 6 & - & 2 & 8 (2x2.8, 2x2.3, 4x1.8 GHz) & Dashcams\\
        \bottomrule
        \end{tabular}
    \label{tab:device_list}
\end{table}

Table~\ref{tab:device_list} lists the devices used in our experiments. We used six smartphones to run \name{} app either as a primary device or external devices and two to emulate inner and outer dashcams. All the devices communicate with each other through the hotspot-enabled primary device using TCP sockets. We categorize devices running \name{} app into strong and weak devices and set up various experimental environments with diverse computing capacities by combining up to three devices in each category as primary and external devices.

We used the TensorFlow Lite task library~\cite{googleTensorFlowLiteTask} with the EfficientDet-lite1~\cite{tensorflowEfficientDetLite42021} object detection model for the outer video analysis. Using them, \name{} either detects potential road hazards or identifies if the driver is tailgating. Non-vehicle objects detected on the road are identified as potential hazards, and vehicle objects that are large enough to indicate that they are very close to the host vehicle are identified for potential tailgating. For the inner video analysis, we used the TensorFlow Lite support library~\cite{thetensorflowauthorsTensorFlowLiteSupport2021} with the MoveNet Lightning~\cite{tensorflowMoveNetLightning2022} pose estimation model to identify the driver's distractedness by examining the driver's hands and eyes.

\begin{table}[h]
\centering
\scriptsize
\setlength{\tabcolsep}{25pt} 
\caption{Parameters and their values.}
        \begin{tabular}{@{}ll@{}}
        \toprule
        Parameter         & Value\\
        \midrule
        Worker sequence length (both inner and outer sequences) & 10\\
        Frame rate control period & 500 $ms$\\
        Network bandwidth (Wi-Fi 802.11ac, 5 GHz) & 100 Mbps\\
        Latency deadline ($L_D$) & 200 $ms$\\
        \bottomrule
        \end{tabular}
    \label{tab:parameters}
\end{table}

Table~\ref{tab:parameters} lists the configurable parameters in \name{} and their values used in all experiments.
Among them, the latency deadline is remarkable. Due to the absence of an objective or standard specification for the deadline of detecting driving hazards and driver distractedness, we provide it as a configurable parameter instead of hard-coding the value in the software. From the user's viewpoint, it can be conceived as a ‘driver guideline’ rather than a ‘deadline’, since it may vary depending on the driving speed. The value can be larger with a low speed because a driver has more time to cope with an event. On the other hand, with a high speed, it should be smaller to provide more time for a driver to handle the event.

\subsection{Performance in the Stable Environments}\label{sec:stable-perf}
We first evaluate the throughput and latency of frame analysis for 30 minutes in a stable environment where the device connection status does not change and device owners do nothing but allow their devices to be fully used for \name{}.

\begin{figure}[t]
     \centering
     \subfigure[Strong device environment with one to three devices.]{
        \includegraphics[width=0.8\columnwidth]{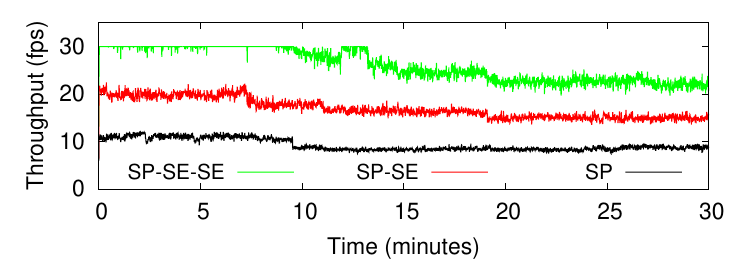}
        \label{fig:throughput-strong}
     }\\
     \subfigure[Weak device environment with one to three devices.]{
        \includegraphics[width=0.8\columnwidth]{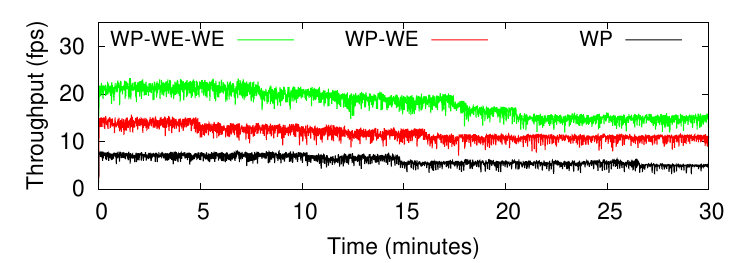}
        \label{fig:throughput-weak}
     }\\
     \subfigure[Heterogeneous device environment with three devices.]{
         \includegraphics[width=0.8\columnwidth]{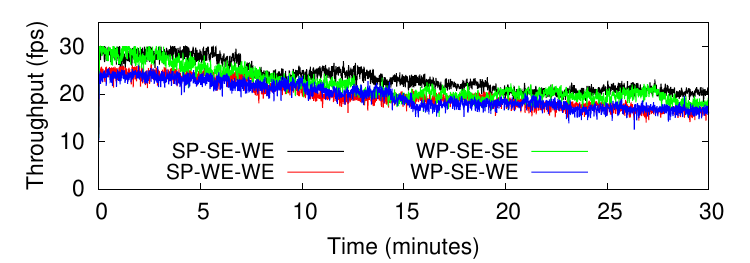}
         \label{fig:throughput-mixed}
     }
     \caption{Frame analysis throughput (frame rate per dashcam).}
     \label{fig:throughput}
\end{figure}

Figure~\ref{fig:throughput} shows the throughput (frame rate per camera) of \name{} under various device environments. SP, WP, SE, and WE denote strong primary, weak primary, strong external, and weak external devices, respectively. \name{} scales well with the number of participating devices, as shown in strong (Figure~\ref{fig:throughput-strong}) and weak (Figure~\ref{fig:throughput-weak}) device environments. As expected, \name{} achieves higher throughput in a strong device environment. The throughput with three heterogeneous devices (Figure~\ref{fig:throughput-mixed}) positions between that with three strong devices (SP-SE-SE) and three weak devices (WP-WE-WE), as easily anticipated. Also, under the heterogeneous device environment, DEVA achieves higher throughput as more strong devices are used. Specifically, the throughput with two strong and one weak devices (SP-SE-WE and WP-SE-SE) is generally higher than that with one strong and two weak devices (SP-WE-WE and WP-SE-WE). The above observations confirm the general scalability of DEVA; it achieves proportional throughput to the aggregated computing power of the participating devices.

One noticeable behavior of \name{} is its gradual throughput degradation as time passes. This is due to the behavior of the thermal governors in Android. As the experiment progresses, the devices' temperature gradually increases and the thermal governor drops CPU frequency accordingly. The default step-wise governor throttles the CPU frequency one step at a time, and the frame rate controller of \name{} adjusts the frame transfer rate of dashcams to the changing CPU frequency, resulting in gradual throughput degradation.

The frame rate is important for the timely detection of driving hazards or driver distractedness in dashcam video analytics since the inter-frame time gap inevitably delays the event/object detection. For example, with a 10 FPS frame rate, up to 100 ms of unavoidable delay occurs in addition to the frame processing latency by \name{}, while only up to 33 ms of delay occurs with a 30 FPS frame rate. As edge devices get stronger, such delay becomes more considerable. The scalability of \name{} enables us to decrease the unavoidable delay and enhance driving safety.

\begin{figure}[t]
     \centering
     \includegraphics[width=0.8\columnwidth]{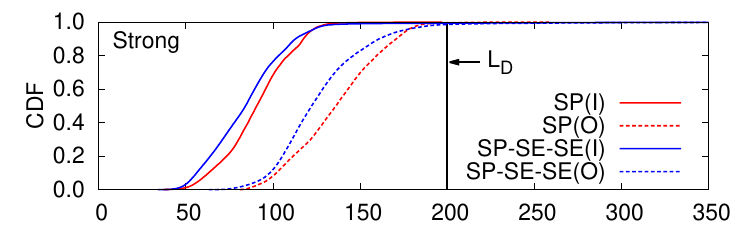}
     \includegraphics[width=0.8\columnwidth]{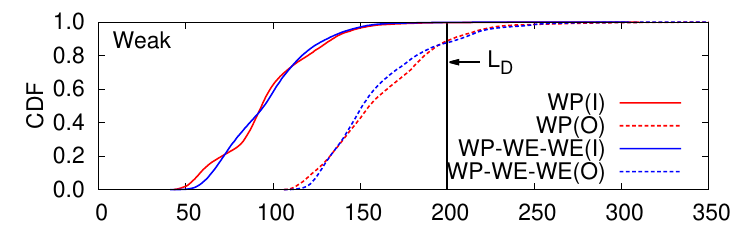}
     \includegraphics[width=0.8\columnwidth]{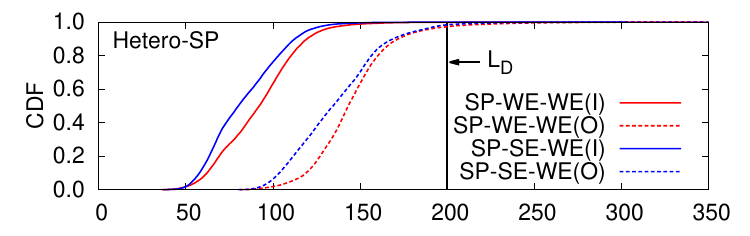}
     \includegraphics[width=0.8\columnwidth]{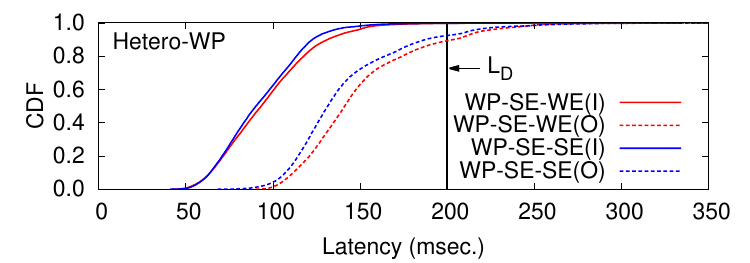}
     \caption{Frame processing latency distribution.}
     \label{fig:latency}
\end{figure}

Figure~\ref{fig:latency} illustrates the cumulative distribution function (CDF) of the frame processing time, which can be interpreted as the proportion of the frame processing latency that does not exceed the value in the $x$-axis. The frame processing time is the elapsed time from when a dashcam starts to transfer a generated frame to the coordinator until its analysis result arrives at the coordinator. The process begins in a dashcam but ends in the primary device, making it difficult to measure the time exactly. Instead of measuring the time as a whole, we divide the time into the frame transfer time from a dashcam to the coordinator and the elapsed time in the coordinator from transferring a frame to a worker to receiving its analysis result from the worker. We measured the latter time in the coordinator and then added the average frame transfer time (outer: 9.28 ms, inner: 8.08 ms) to incorporate the former time.

In all cases, we can see the large latency difference between inner and outer frames, which shows the inter-video workload heterogeneity due to the noticeably different computational complexity of analysis models used for inner and outer videos. Specifically, the average image analysis times for inner and outer frames were 43 ms and 110 ms, respectively, in a weak device, and 31.5 ms and 82.7 ms, respectively, in a strong one.

One may have a concern about the lower latency of the SP-SE-SE case than the SP case although the former incorporates frame transfers from the primary device to external devices which the latter does not. It is due to two reasons. First, in the primary device, frames assigned to itself are also delivered through socket communication to worker 0, so that \name{} does not need to differentiate worker 0 and other workers. Hence, the frames also pass through the network protocol stack bypassing only the physical communication. Considering the average frame transfer time, which includes the protocol stack latency, is less than 10 ms, the physical communication time takes only a marginal portion of the entire frame processing latency. Second, since the primary device acts both as a worker and the coordinator, the coordinating jobs that impose time-varying loads on the device frequently defer the image analysis tasks, increasing their overall processing time. And since such deferred frames occur only in the primary device, as we use more (external) devices in addition to the primary device, the total number of processed frames increases while the number of deferred frames remains the same. Hence, as we use more devices, the ratio of the deferred frames to the entire processed frames decreases, lowering the overall latency distribution. The reduced overall latency outweighs the physical communication time in the SP-SE-SE case, finally letting the SP-SE-SE case show a lower latency distribution than the SP case. Interestingly, however, such a result does not occur in the weak device environment. This is due to the largely varying image analysis times in the weak devices (Figure 8), which incurs large overall processing times for numerous frames. Due to such frames, the deferred frames in the primary device have less impact on the overall latency distribution, leading to similar distributions between WP and WP-WE-WE cases.

Although \name{} tries to satisfy the latency guideline by controlling the frame rate, it occasionally fails to do so because the frame rate controller uses the \textit{average} image analysis and frame transfer times when determining the frame rate. Noting that the real image analysis and frame transfer times can vary by a large margin, there can be frames whose processing times are far larger than the estimated value, finally missing the latency deadline, as shown in the weak and heterogeneous device environments in Figure~\ref{fig:latency}.

Figure~\ref{fig:analysis} confirms the argument by illustrating the image analysis times of every processed outer frame. The analysis time of each frame largely varies in weak devices, sometimes exceeding the latency deadline by itself, resulting in a larger deadline miss ratio than in strong devices. Since the average image analysis time in a weak device is also larger than that in a strong device, the overall frame processing latency increases as more weak devices are used (e.g., compare SP-WE-WE(O) with SP-SE-WE(O) in Figure~\ref{fig:latency}). Hence, more frames miss the deadline in the weak device environment than in the strong device environment.

\begin{figure}[t]
     \centering
     \hspace{-2mm}\includegraphics[width=0.48\columnwidth]{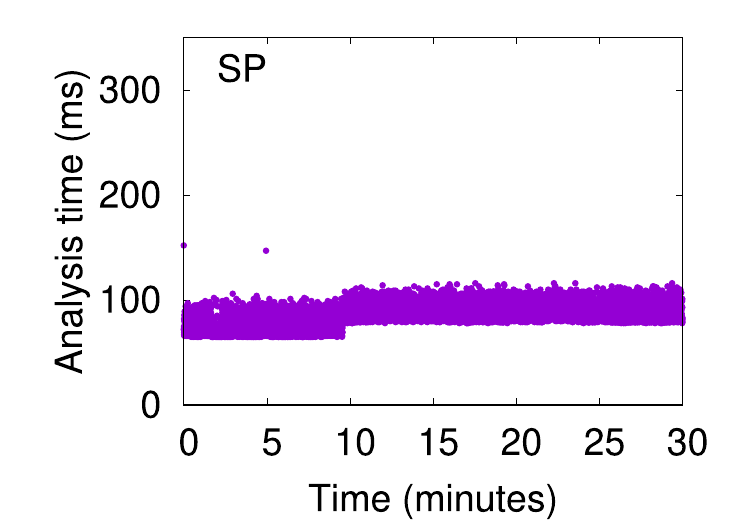}
     \hspace{-3mm}\includegraphics[width=0.48\columnwidth]{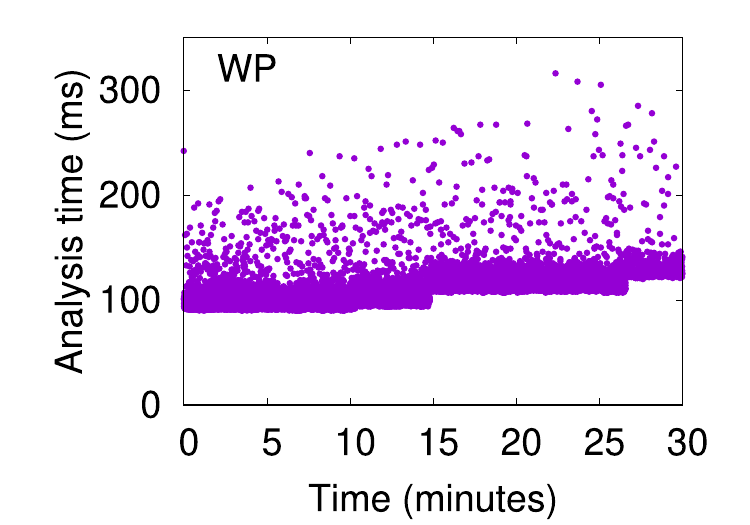}\\
     \hspace{-2mm}\includegraphics[width=0.48\columnwidth]{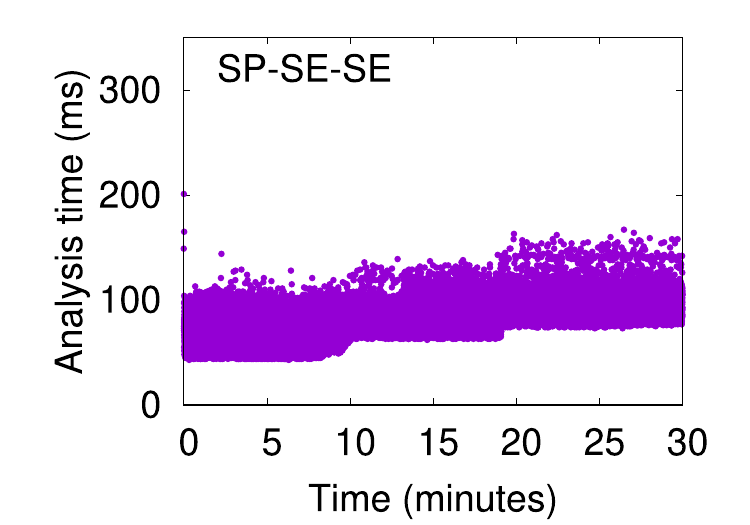}
     \hspace{-3mm}\includegraphics[width=0.48\columnwidth]{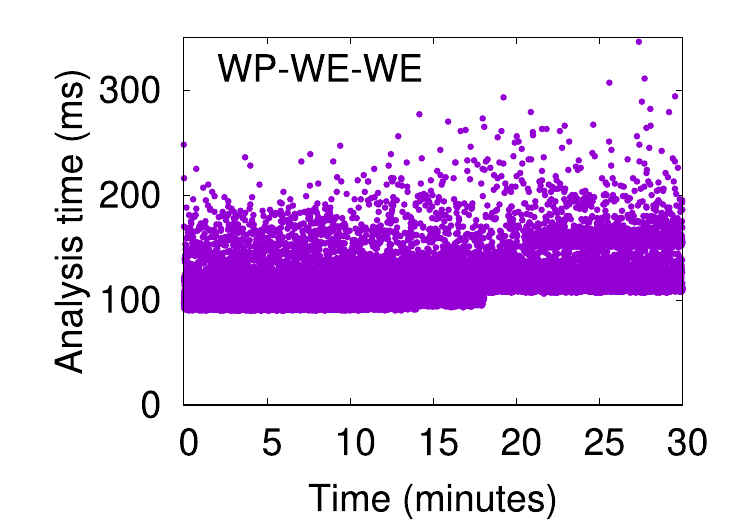}
     \caption{Image analysis times of the processed outer frames.}
     \label{fig:analysis}
\end{figure}

More importantly, the deadline miss ratio highly depends on the power of the \textit{primary} device; more frames miss the deadline when using a weak primary device. For example, comparing SP-WE-WE and WP-SE-WE cases that use the same device set (one strong and two weak devices) in the heterogeneous environment, the WP-SE-WE case shows far more deadline missed frames than the SP-WE-WE case while their throughput does not show a noticeable difference (Figure~\ref{fig:throughput-mixed}). The result confirms the importance of the primary device which takes both the coordinator and worker roles in \name{}. Due to the coordination tasks, the image analysis times in the primary device become large, making more frames miss the deadline. Therefore, among the available devices, it is necessary to use the most powerful device as the primary device.

Interestingly, the number of devices does not noticeably affect the deadline miss ratio, as shown in the strong and weak device environments in Figure~\ref{fig:latency}. This is due to the dynamic frame rate control scheme in \name{} which steadily maintains the in-worker frame queue lengths by adjusting the overall frame rate following the number of devices. The scheme also takes a critical role in preventing \name{} from abnormally stopping due to the frame buffer overflow. As shown in Figure~\ref{fig:throughput-strong}, less than three strong devices do not have sufficient aggregated computing power to analyze two full-frame videos. Without the scheme (i.e., dashcams transfer frames at a constant, 30 fps rate), the frame analysis tasks are gradually delayed due to insufficient computing power. The delayed frames are accumulated in the frame buffer, finally exceeding the buffer capacity and impairing the entire system. Hence, to analyze real-time videos using a small number of resource-constrained mobile devices as DEVA aims to, not only the frame scheduling but also the frame rate control schemes are indispensable.\footnote{This is why we could not compare \name{} with other schemes. To the best of our knowledge, DEVA is the first system that equips both frame scheduling and frame rate control functionalities for real-time distributed video analytics.}

To confirm the above argument, we experimented using an existing task scheduling algorithm for mobile devices, work-stealing \cite{fernandoMobileCrowdComputing2012}. The algorithm does not accompany a frame rate control scheme since it is not designed for video analytics applications. We implemented the algorithm into \name{} such that each worker receives two frames from the coordinator (i.e., \textit{steal} two tasks) whenever its frame queue is empty, disabling the frame rate control functionality.

\begin{figure}[t]
     \centering
     \includegraphics[width=0.8\columnwidth]{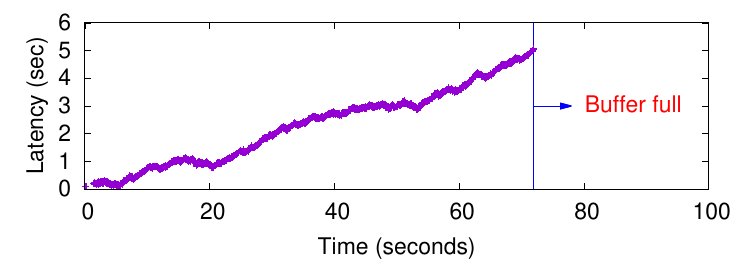}
     \caption{Latency variation with the work-stealing task scheduling algorithm.}
     \label{fig:steal}
\end{figure}

Figure~\ref{fig:steal} shows the outer frame processing latency variation as time passes using two strong devices (SP-SE). From the early stage of the experiment, frames begin to be delayed in the coordinator's frame queue, showing unacceptable and increasing latency according to the growing queue length. At 72 seconds, the delayed frames filled up the entire frame buffer that has 300 frames of capacity, and the system stopped due to the buffer overflow. Such a phenomenon occurred in every combination of devices, although the entire system stopped at different times in each case. The result clearly confirms the necessity of the frame rate control scheme for distributed real-time video analytics across resource-constrained mobile devices.

\subsection{Performance in the Unstable Environments}\label{sec:unstable-perf}
In this section, we evaluate the throughput and latency of frame analysis for 30 minutes in unstable environments, where the device connection status or the amount of owner-directed jobs varies, to evaluate \name{} in more realistic scenarios.

\subsubsection{Performance with Device Join or Leave}\label{sec:jl}

We first examine how \name{} behaves when a new device joins or an existing device leaves the system while it is working. We focus on how quickly \name{} adjusts to the connection status changes by updating its frame rate accordingly, keeping low latency and high throughput of frame analysis. We assume two scenarios of the connection status changes; two-phase change and random change scenarios. In the two-phase change scenario, external devices continuously join the system (increase phase) and then continuously leave the system (decrease phase). In the random change scenario, external devices randomly join or leave the system. In both scenarios, a device connection status change occurs every one minute. We assume that the primary device, possibly the driver's, never leaves the system.

\begin{figure}[t]
     \centering
     \subfigure[Throughput variation with Two-phase Join/Leave.]{
        \includegraphics[width=0.8\columnwidth]{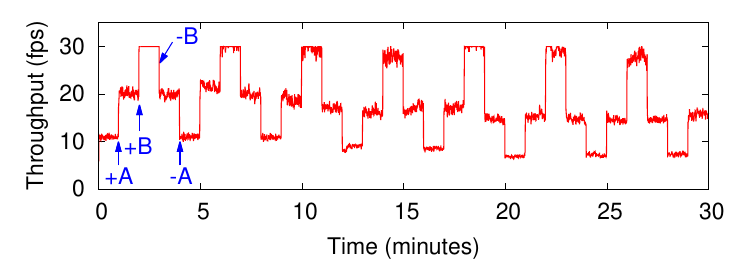}
        \label{fig:throughput-jl-2phase}
     }\\
     \subfigure[Throughput variation with Random Join/Leave.]{
        \includegraphics[width=0.8\columnwidth]{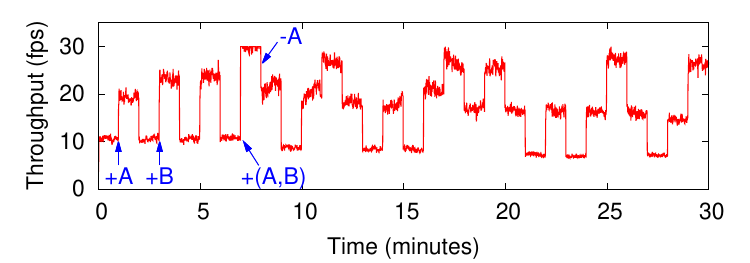}
        \label{fig:throughput-jl-random}
     }\\
     \subfigure[Latency distribution.]{
        \includegraphics[width=0.8\columnwidth]{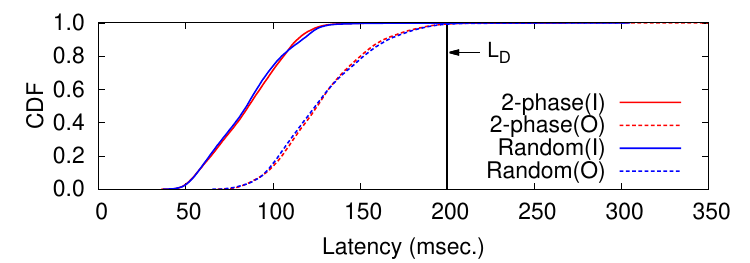}
        \label{fig:latency-jl}
     }
     \caption{Throughput and latency with device join/leave.}
     \label{fig:join-leave}
\end{figure}

Figure~\ref{fig:join-leave} shows the results in the strong device environment (SP-SE-SE), i.e., ever-connected SP and two dynamically joining or leaving SEs (i.e., A and B). In both scenarios, \name{} quickly adjusts to the varying device connection status owing to its per-device queueing model-based frame rate control. Since \name{} estimates the entire system capacity by summing up the individual device capacity, the capacity of a joining or leaving device can be immediately reflected in the whole system capacity upon recognizing such an event, resulting in the step-wise (not gradual) throughput variation.

In Figures~\ref{fig:throughput-jl-2phase} and \ref{fig:throughput-jl-random}, ``+A" and ``-A" denote the join and leave of device A, respectively. In both scenarios, the throughput increases more when device B joins than when device A joins since device B (Galaxy S22) is stronger than A (Find X2) as shown in Table~\ref{tab:device_list}. In the two-phase change scenario, the `+A, +B, -B, and -A' change pattern repeats, and the throughput shows homogeneous patterns accordingly except for the gradual degradation due to the thermal control of Android. In the random change scenario, two devices join simultaneously at time 7, achieving the maximum frame rate, and then device A leaves the system at time 8. We can observe the gradual throughput degradation also in this scenario. As such, \name{} promptly adjusts the frame rate, incorporating not only the device join/leave events but also the individual device capacity and thermal control.

Figure~\ref{fig:latency-jl} confirms that the worker frame queue lengths are maintained stable despite the connection status changes, and thus most frames meet the latency guideline as was in the same SP-SE-SE case without connection status changes (Figure~\ref{fig:latency}). This is one of the important characteristics of \name{}. When a device leaves the system, the worker queue lengths of the remaining devices may rapidly grow, if the frame rate controller does not immediately decrease the frame rate, resulting in increased frame processing latency and latency-missing frames. By rapidly adjusting the frame rate to the changing device connection status, \name{} could effectively maintain the worker queue lengths stable.

\subsubsection{Performance with User Interactions}

Now, we examine how \name{} behaves when external device owners interact with their devices such as launching a new app and finishing a running app. As in the previous experiments, we focus on how rapidly \name{} responds to the changing computing capacity due to user interactions by adjusting both the frame rate and frame distribution. For the ease of experiments, we emulated user interactions inside the \name{} app by starting or finishing a thread that executes some random image analysis tasks independently. We used the two-phase change scenario in Section~\ref{sec:jl}, replacing device join and leave with the start and stop of the thread, respectively.

\begin{figure}[t]
     \centering
     \subfigure[Throughput variation with user interactions.]{
        \includegraphics[width=0.8\columnwidth]{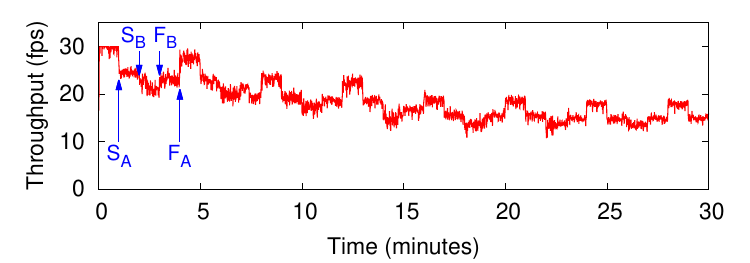}
        \label{fig:throughput-ui}
     }\\
     \subfigure[Latency distribution with user interactions.]{
        \includegraphics[width=0.8\columnwidth]{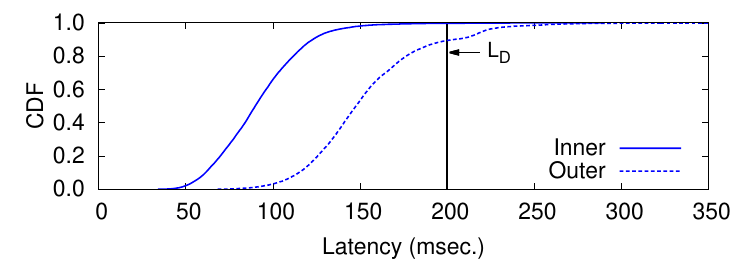}
        \label{fig:latency-ui}
     }
     \caption{Throughput and latency with user interactions.}
     \label{fig:ui}
\end{figure}

Figure~\ref{fig:ui} shows the throughput and latency of \name{} in the presence of user interactions. We used two strong and one weak devices (SP-SE-WE) where user interactions occur only in external devices assuming that the owner of the primary device (possibly, the driver) does not interact with the device. In Figure~\ref{fig:throughput-ui}, $S_A$ and $F_A$ denote the start and finish of the user interaction in device A. We used Galaxy S22 and Pixel 5 for external devices A (strong) and B (weak), respectively.

User interactions affect the image analysis time ($T_C^i$) by sharing the computing resources with \name{}. When a user-directed job starts running, the average image analysis time increases, and the frame rate controller in \name{} decreases the frame rate accordingly. When the job finishes, \name{} increases the frame rate again. Figure~\ref{fig:throughput-ui} shows such frame rate variations. We can see that \name{} rapidly reflects the changed available computing capacity to the frame rate. The increased image analysis time due to the user-directed job means the decreased device power to \name{}. In other words, it has the same effect as using weaker devices. Hence, the overall frame processing time increases and more frames miss the latency deadline. As shown in Figure~\ref{fig:latency-ui}, about 10\% of the outer frames miss the latency deadline, which is larger than the deadline miss ratio in the same SP-SE-WE case without user interaction (Figure~\ref{fig:latency}). Nonetheless, the latency rarely exceeds 250 ms.

\subsection{Detailed Analysis of \name{}}
In this section, we investigate the behavior of \name{} in detail to analyze its performance characteristics.

\subsubsection{Worker Queue Length Variation}
Given the large image analysis time, the frame queue lengths in workers greatly affect the overall frame processing latency. For example, considering the average image analysis time for outer frames in a weak device is 110 ms in our experiments, an incoming frame to the queue may miss the latency deadline even when only one outer frame exists in the queue. Hence, maintaining a small (almost zero) and stable queue length is paramount to meeting the latency deadline.

\begin{figure}[t]
        \centering
        \hspace{-2mm}\includegraphics[width=0.48\columnwidth]{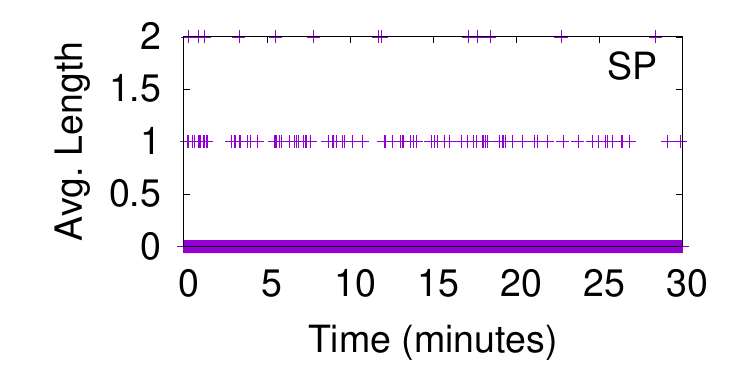}
        \hspace{-3mm}\includegraphics[width=0.48\columnwidth]{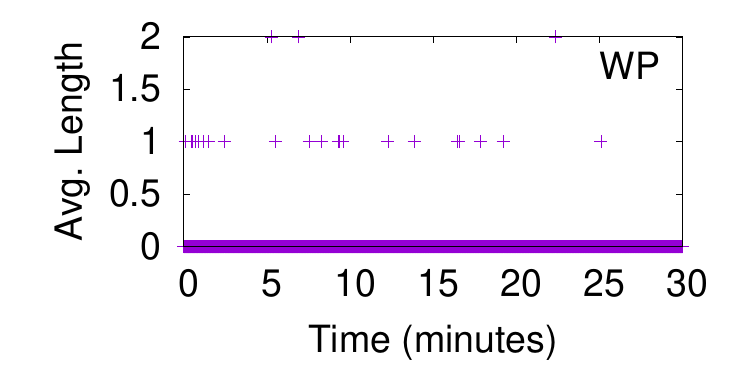}

        \hspace{-2mm}\includegraphics[width=0.48\columnwidth]{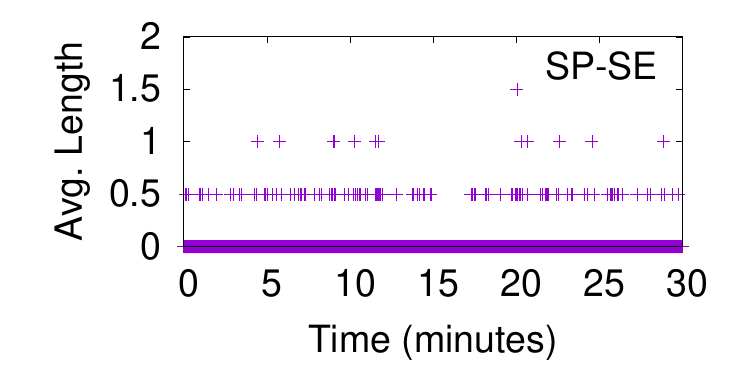}
        \hspace{-3mm}\includegraphics[width=0.48\columnwidth]{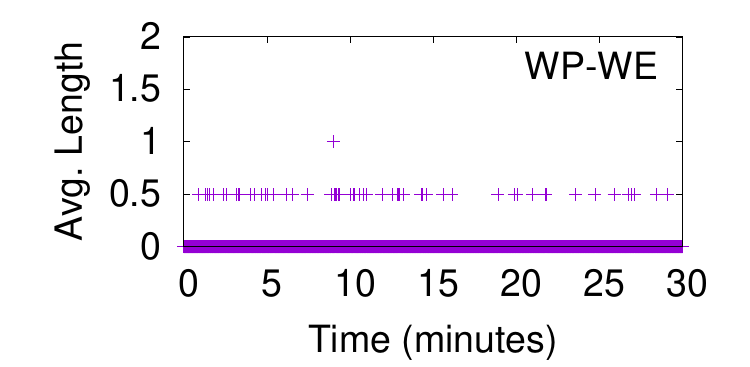}

        \hspace{-2mm}\includegraphics[width=0.48\columnwidth]{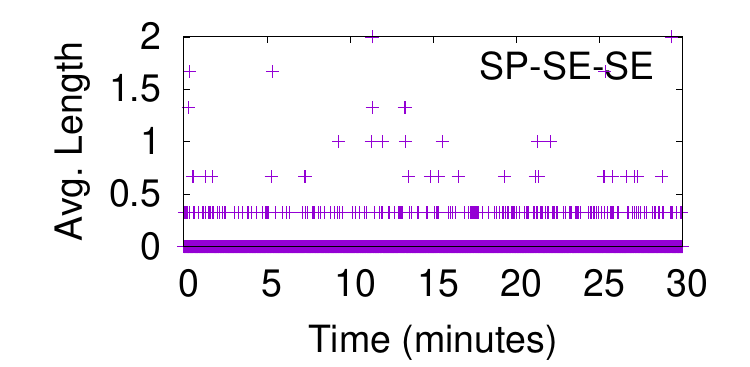}
        \hspace{-3mm}\includegraphics[width=0.48\columnwidth]{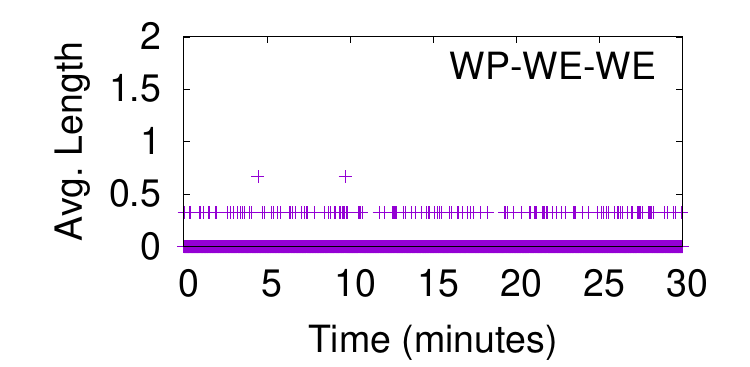}

        \hspace{-2mm}\includegraphics[width=0.48\columnwidth]{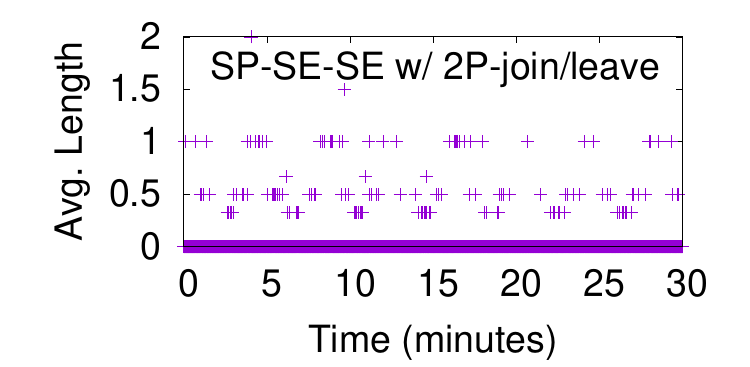}
        \hspace{-3mm}\includegraphics[width=0.48\columnwidth]{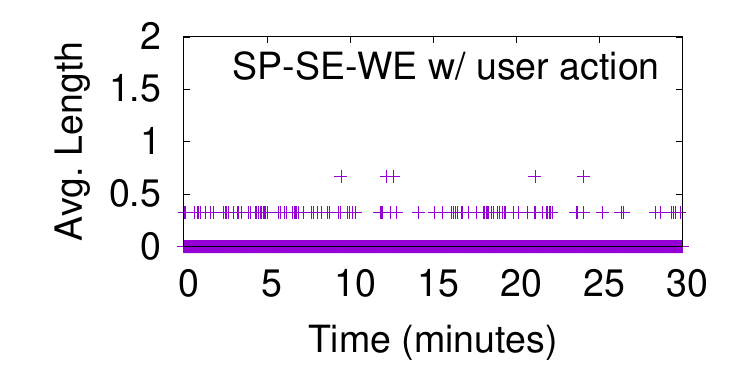}
     \caption{Average length of the worker frame queues.}
     \label{fig:queue}
\end{figure}

Figure~\ref{fig:queue} shows the average worker queue lengths in the stable (upper six figures) and unstable (lower two figures) environments. We can see that \name{} maintains the worker queue lengths mostly to zero in all cases, thus achieving a high deadline satisfaction ratio. Even in unstable environments, \name{} effectively suppresses the excessive queue length increase through the prompt frame rate adjustment.

One can be concerned about numerous nonzero points in the strong device environment (SP, SP-SE, and SP-SE-SE), which makes their near-zero deadline miss ratios shown in Figures~\ref{fig:latency} and \ref{fig:latency-jl} questionable. The answer is in the parallel analysis of multiple frames on multicores. In our experiments, \name{} fully exploits eight cores in each device by creating two concurrent image analysis threads, each using four TensorFlow inference threads. Owing to the parallel image analysis functionality, even when a frame already exists in the queue, an incoming frame does not necessarily wait until the analysis of the existing frame is completed. However, with more than one frame in the queue, the incoming frame cannot help waiting for the completion of one of the frames, possibly missing the latency deadline. \name{} achieves near-zero deadline miss ratios in the strong device environment by keeping the average queue lengths not exceeding one in most cases.

The number of nonzero points is smaller in the weak device environment (WP, WP-WE, and WP-WE-WE) than in the strong device environment due to the low frame rate. However, far more frames miss the deadline in the weak device environment as shown in Figure~\ref{fig:latency}, despite the parallel image analysis. This is due to the large average and deviation of the image analysis times in weak devices, as shown in Figure~\ref{fig:analysis}, which makes frames with large analysis times miss the deadline even when the queue length is zero. Such a case occurs specifically when all the image analysis threads just started analyzing preceding frames, and thus the incoming frame to the empty queue has to wait for a long time.

\subsubsection{Frame Distribution among Workers}
The frame scheduler in \name{} distributes frames across workers based on worker weights, reflecting the real-time computing capacity of workers. In a stable environment and also in an unstable environment with joining or leaving devices, it is designed to distribute frames in proportion to the native computing power of each device, although slight, instantaneous variations may occur depending on the real-time states of devices. With user interactions imposed, it also incorporates the changing computing capacity due to user-directed jobs when distributing frames.

\begin{figure}[t]
    \centering
    \includegraphics[width=0.8\columnwidth]{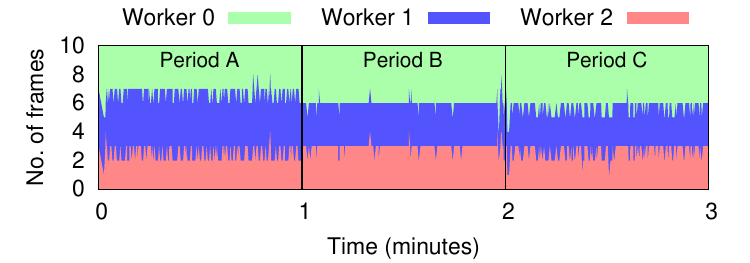}
    \caption{Frame distribution with user interactions in the SP-SE-WE environment. Period A has no interaction. Period B has interaction occurring in a strong external device. Period C has interactions occurring in both strong and weak external devices. Workers 0, 1, and 2 are in the strong primary, strong external, and weak external devices, respectively. }
    \label{fig:frame-dist}
\end{figure}

Figure~\ref{fig:frame-dist} shows the number of assigned frames to each worker whenever a new inner or outer worker sequence (length = 10) is generated in the experiment for Figure~\ref{fig:ui}. For better visibility, we focus on a short interval (0\myapprox3 minutes) which includes every kind of user interaction. In period A, frames are distributed in a 3.40:3.99:2.61 ratio, on average, which well reflects the native computing capacity of each device. In period B, the average ratio changes to 3.99:3.01:3.00, reflecting the reduced available computing capacity by the user interaction in the strong external device where worker 1 resides. In period C, the frame scheduler assigns fewer frames to workers in both external devices, resulting in a 4.29:3.12:2.59 ratio, reflecting user interactions in the devices.

\section{Lessons Learned}\label{sec:lesson}
In this section, we present meaningful lessons we have learned while evaluating the proposed system.
\begin{itemize}[leftmargin=*]
   \item \textbf{Number of devices v.s. individual device power}\\
   As shown in Figures~\ref{fig:throughput} and \ref{fig:latency}, the number of devices largely contributes to the frame processing throughput but has less impact on the latency. On the contrary, the individual device power greatly affects the frame processing latency since the image analysis time that depends on the CPU power takes a large portion of the latency. In terms of latency, the individual device power is more important than the number of devices. However, since a low frame rate incurs a large unavoidable delay incurred by the inter-frame time gap, it is necessary to maintain a sufficiently large frame rate for all video sources. As we use more video sources, the per-source frame rate decreases, necessitating more devices to maintain the frame rate. In conclusion, the device's power matters with a small number of dashcams while both the number and power of devices are important with many dashcams.
   \item \textbf{Internal device network matters}\\
   \name{} consumes a large amount of network bandwidth due to the frame-unit video transfer. For example, the video files we used for experiments have 116 KB (outer) and 101 KB (inner) of average frame sizes after compression, which consumes approximately 52 Mbps of network bandwidth (assuming 30 FPS) to be transferred from dashcams to the coordinator. Considering that frames assigned to external workers need to be transferred again to external devices, \name{} consumes up to 87 Mbps of bandwidth for video transfers when using two external devices, which imposes a heavy load on the network constructed via Wi\nobreakdash-Fi Direct or hotspot. Due to the heavy load, we sometimes observed abnormal behaviors of the primary device that relays the videos, e.g., stopping data transfer to a specific external device for a while and resuming it. Such phenomena occurred when we used Wi\nobreakdash-Fi 5 (802.11ac) with 2.4 GHz frequency but mostly disappeared when using 5 GHz frequency which provides a higher speed than 2.4 GHz. Noting that Wi\nobreakdash-Fi 5 is known to support typically higher speed than 87 Mbps even with 2.4 GHz within a short range as in the vehicle, our observation reveals that sufficiently higher bandwidth is required for the stable operation of \name{}.
\end{itemize}

\section{Future Work}
We plan to further enhance the proposed system by incorporating two additional functionalities into \name{}.

\begin{itemize}[leftmargin=*]
\item \textbf{Thermal Control-aware Analysis Model Selection}\\
As shown in all experimental results, the thermal governor in Android noticeably drops the CPU frequency as the device is heated, which directly affects the image analysis time. After a long execution of \name{}, the average image analysis time increases by about 30\% due to the decreased CPU frequency (Figure~\ref{fig:analysis}). And the frame rate decreases accordingly, showing practically unacceptable rates (about 4\myapprox5 FPS per dashcam) when using a single weak device (Figure~\ref{fig:throughput-weak}).

As a potential solution, we may install multiple image analysis models with different complexities from each other and dynamically select one considering the CPU frequency. For example, the object detection and pose estimation models we used in \name{} have other sibling models with various complexities. We leave it as one of our future work.

\item \textbf{Overdue Frame Dropping}\\
Video analysis for driving safety requires timely analysis of each frame without excessive delay. Although the weighted frame scheduling and dynamic frame rate control in \name{} uses the estimated computing capacity to minimize the overall frame analysis latency, there is an inevitable time gap from the computing capacity changes to its remedy becomes effective. Also, the estimated computing capacity may not accurately represent the real capacity. Due to these reasons, frame queues in workers can have overdue frames whose frame processing latency cannot meet the latency deadline and thus the analysis results would be useless. Processing these frames unnecessarily consumes system resources and more importantly, it delays processing upcoming frames, accumulating overdue frames in the queue for a while.

As a potential solution, we may implement workers to identify and discard overdue frames without further processing. Specifically, whenever a worker $w_i$ dequeues a frame $f$ from its frame queue for analysis, it measures the elapsed time of the frame in the queue ($T_q^f$) and estimates the overall latency of the frame by
$T^f = 2T_F + T_q^f + T_A +T_R$, where $T_A$ is either $T_I^i$ (if $f$ is an inner frame) or $T_O^i$ (otherwise). If $T^f$ exceeds the latency deadline ($L_D$) then it discards the frame without analysis. However, such frame dropping may, in turn, introduce another problem. The missing performance index records for the dropped frames may lead to an incorrect estimation of the worker computing capacity. While the existence of dropped frames in a worker indicates the overloaded state of the worker, the coordinator cannot recognize it due to the missing performance index records for the frames. Hence, workers also need to send their performance index records to the coordinator, embedding precise information about the worker state, so that the coordinator can use it for the worker weight calculation. We leave it as another future work.
\end{itemize}

\section{Conclusion}\label{sec:conclusion}
In this paper, we have addressed issues of real-time dashcam video analytics using resource-constrained in-vehicle edge devices. As a solution, we have developed \name{} that provides stable and low-latency video analytics for the timely detection of road hazards and driver distractedness by exploiting consolidated computing resources across one or more edge devices. \name{} incorporates several techniques including pipelined task execution, weighted frame scheduling, and dynamic frame rate control to achieve its goal even under dynamically changing resource provisioning environments. We have implemented \name{} in an Android app, along with a dashcam emulation app that can be used in vehicles without dashcams. Experimental results using real-world datasets under realistic scenarios showed the feasibility of real-time dashcam video analytics for driving safety, providing an effective use case of dashcam video data that is otherwise discarded without any use.

\bibliographystyle{elsarticle-num}
\bibliography{references}

\end{document}